\documentclass[10pt]{article}
\pdfoutput=1
\usepackage{mmap,amsmath,amsthm,amsfonts,
textalpha,longtable,authblk}

\usepackage[A4]{vmargin}
\setmarginsrb{25mm}{25mm}{25mm}{24mm}{0pt}{0pt}{12pt}{11mm}

\usepackage[arrow,matrix]{xy}
\let\ar\relax

\usepackage[unicode=true,colorlinks=true,linkcolor=blue,
citecolor=blue,urlcolor=blue]{hyperref}

\DeclareMathOperator{\GL}{GL}
\DeclareMathOperator{\dom}{dom}
\DeclareMathOperator{\ar}{ar}
\DeclareMathOperator{\Tm}{Tm}
\DeclareMathOperator{\Pro}{Pr}
\DeclareMathOperator{\negl}{negl}
\DeclareMathOperator{\supp}{supp}
\DeclareMathOperator{\SLP}{SLP}
\DeclareMathOperator{\Smpl}{Smpl}
\DeclareMathOperator{\CNOT}{CNOT}

\newcommand{\B}{\{0,1\}}
\newcommand{\latin}{\textit}
\newcommand{\pd}{\mathcal}
\newcommand{\cl}{\mathfrak}
\newcommand{\fm}{\mathtt}
\newcommand{\cs}{,\,}
\newcommand{\st}{\,|\,}
\newcommand{\rep}[2]{[#1]_{#2}}
\newcommand{\id}{\mathrm{id}}
\newcommand{\alg}[1]{\langle#1\rangle}
\newcommand{\pns}{\vartriangleleft}
\newcommand{\acl}[1]{\lvert#1\rvert}
\newcommand{\N}{\mathbb N}
\newcommand{\Z}{\mathbb Z}
\newcommand{\R}{\mathbb R}
\newcommand{\suz}{\mathbb}
\newcommand{\UZ}{\Z^\star}
\newcommand{\rv}{\mathbf}
\newcommand{\rvd}{\sim}
\newcommand{\Prob}[1]{\Pro[#1]}
\newcommand{\ket}[1]{\lvert#1\rangle}
\newcommand{\nreg}{\underline}
\newcommand{\Forall}[1]{\forall\,#1\ }

\newcommand{\impl}{\xyarrow@{=>}}
\newcommand{\Dk}{\text{if }\Forall k(\supp\pd D_k\subseteq D_k)}

\newcommand{\headPsi}{\texorpdfstring
{$\mathbf{\Psi}$}{\textPsi}}
\newcommand{\headOmega}{\texorpdfstring
{$\mathbf{\Omega}$}{\textOmega}}
\newcommand{\sectionheadOmega}{\texorpdfstring
{$\protect\sectionheadmathbf{\Omega}$}{\textOmega}}
\newcommand{\headclV}{\texorpdfstring
{$\boldsymbol{\cl V}$}{V}}

\newcommand{\ORCID}{0000-0002-3960-3867}
\newcommand{\keywords}{post-quantum cryptography, universal
algebra, expanded group, family of computational universal
algebras, black-box model, weakly pseudo-free family}

\let\sectionheadmathbf\mathbf

\newenvironment{arenum}{\begin{list}
{\textup{\arabic{enumi}.}}{\usecounter{enumi}
\settowidth{\labelwidth}{\textup{0.}}
\addtolength{\leftmargin}{.7\parindent}
}}{\end{list}}

\newenvironment{roenum}{\begin{list}
{\textup{(\roman{enumi})}}{\usecounter{enumi}
\settowidth{\labelwidth}{\textup{(iii)}}
\addtolength{\leftmargin}{.7\parindent}
}}{\end{list}}

\theoremstyle{plain}
\newtheorem{theorem}{Theorem}[section]
\newtheorem{lemma}[theorem]{Lemma}
\newtheorem{proposition}[theorem]{Proposition}

\theoremstyle{definition}
\newtheorem{definition}[theorem]{Definition}
\newtheorem{example}[theorem]{Example}

\theoremstyle{remark}
\newtheorem{remark}[theorem]{\bfseries Remark}

\title{There Are No Post-Quantum Weakly Pseudo-Free Families
in Any Nontrivial Variety of Expanded Groups}

\author{Mikhail Anokhin}

\affil{Information Security Center\\Faculty of Computational
Mathematics and Cybernetics\\Lomonosov University (a.k.a.\
Lomonosov Moscow State University)\\Moscow, Russia

\smallskip

E-mail address:~\texttt{anokhin@mccme.ru}

ORCID:~\href{https://orcid.org/\ORCID}{\ORCID}}

\date{}

\makeatletter\hypersetup{pdftitle=\@title,pdfauthor={Mikhail
Anokhin},pdfkeywords=\keywords}\makeatother

\begin{document}
{\renewcommand{\thefootnote}{}\footnotetext{This is an
extended version of the paper (with the same title)
published in \textit{International Journal of Algebra and
Computation}, 34(4):471--490, 2024.}}

\maketitle

\begin{abstract}
Let $\Omega$ be a finite set of finitary operation symbols
and let $\cl V$ be a nontrivial variety of
$\Omega$-algebras. Assume that for some set
$\Gamma\subseteq\Omega$ of group operation symbols, all
$\Omega$-algebras in $\cl V$ are groups under the operations
associated with the symbols in~$\Gamma$. In other words,
$\cl V$ is assumed to be a nontrivial variety of expanded
groups. In particular, $\cl V$ can be a nontrivial variety
of groups or rings. Our main result is that there are no
post-quantum weakly pseudo-free families in $\cl V$, even in
the worst-case setting and/or the black-box model. In this
paper, we restrict ourselves to families $(H_d\st d\in D)$
of computational and black-box $\Omega$-algebras (where
$D\subseteq\B^*$) such that for every $d\in D$, each element
of $H_d$ is represented by a unique bit string of length
polynomial in the length of~$d$. In our main result, we use
straight-line programs to represent nontrivial relations
between elements of $\Omega$-algebras. Note that under
certain conditions, this result depends on the
classification of finite simple groups. Also, we define and
study some types of weak pseudo-freeness for families of
computational and black-box $\Omega$-algebras.

\medskip

\end{abstract}

{\renewcommand{\mathbf}{}\renewcommand{\boldsymbol}{}
\tableofcontents}

\section{Introduction}

Let $\Omega$ be a finite set of finitary operation symbols
and let $\cl V$ be a variety of $\Omega$-algebras. (See
Subsection~\ref{ss:univalgprelim} for definitions.)
Informally, a family of computational $\Omega$-algebras is a
family of $\Omega$-algebras whose elements are represented
by bit strings in such a way that equality testing, the
fundamental operations, and generating random elements can
be performed efficiently. This means that each element $h$
of every $\Omega$-algebra in the family is assigned one or
more bit strings, called representations of~$h$. Efficient
algorithms performing the above-mentioned operations work
with these representations of elements; moreover, these
algorithms take as input the index of the $\Omega$-algebra
in the family. Loosely speaking, a family of computational
$\Omega$-algebras is called pseudo-free in $\cl V$ if all
members of this family belong to $\cl V$ and, given a random
member $H$ of the family (for a given security parameter)
and random elements $g_1,\dots,g_m\in H$ (where $m\ge1$), it
is computationally hard to find a system of equations
\begin{equation}
\label{e:syseq}
v_i(a_1,\dots,a_m;x_1,\dots,x_n)
=w_i(a_1,\dots,a_m;x_1,\dots,x_n),\quad i\in\{1,\dots,s\},
\end{equation}
in the variables $x_1,\dots,x_n$ together with elements
$h_1,\dots,h_n\in H$ such that
\begin{itemize}
\item for each $i\in\{1,\dots,s\}$,
$v_i(a_1,\dots,a_m;x_1,\dots,x_n)$ and
$w_i(a_1,\dots,a_m;x_1,\dots,x_n)$ are elements of the $\cl
V$-free $\Omega$-algebra freely generated by
$a_1,\dots,a_m,x_1,\dots,x_n$,

\item system~\eqref{e:syseq} is unsatisfiable in the $\cl
V$-free $\Omega$-algebra freely generated by
$a_1,\dots,a_m$, and

\item $v_i(g_1,\dots,g_m;h_1,\dots,h_n)
=w_i(g_1,\dots,g_m;h_1,\dots,h_n)$ in $H$ for all
$i\in\{1,\dots,s\}$.
\end{itemize}
If a family of computational $\Omega$-algebras satisfies
this definition with the additional requirement that $n=0$
(i.e., that the equations in~\eqref{e:syseq} be
variable-free), then this family is said to be weakly
pseudo-free in~$\cl V$. By fixing the number $s$ of
equations in the definition of a pseudo-free (resp., weakly
pseudo-free) family in $\cl V$, we obtain a definition of an
$s$-pseudo-free (resp., weakly $s$-pseudo-free) family
in~$\cl V$. Of course, pseudo-freeness (in any above
version) may depend heavily on the form in which
system~\eqref{e:syseq} is required to be found, i.e., on the
representation of such systems.

The notion of pseudo-freeness (which is a variant of weak
$1$-pseudo-freeness in the above sense) was introduced by
Hohenberger in~\cite[Section~4.5]{Hoh03} for black-box
groups. Rivest gave formal definitions of a pseudo-free
family of computational groups
(see~\cite[Definition~2]{Riv04}, \cite[Slide~17]{Riv04pres})
and a weakly pseudo-free one
(see~\cite[Slide~11]{Riv04pres}). These authors consider
(weak) pseudo-freeness only in the varieties of all groups
and of all abelian groups. Note that pseudo-freeness (resp.,
weak pseudo-freeness) in~\cite{Riv04, Riv04pres} is in fact
$1$-pseudo-freeness (resp., weak $1$-pseudo-freeness) in our
terminology. For motivation of the study of pseudo-freeness,
we refer the reader to~\cite{Hoh03, Riv04, Mic10}.

Suppose $\fm H=(H_d\st d\in D)$ is a family of computational
$\Omega$-algebras, where $D\subseteq\B^*$. (We specify only
the $\Omega$-algebras here.) Then this family is said to
have exponential size if there exists a polynomial $\xi$
such that $\acl{H_d}\le2^{\xi(\acl d)}$ for all $d\in D$
(see also~\cite[Definition~3.2]{Ano21}). The family $\fm H$
is called polynomially bounded if there exists a polynomial
$\eta$ such that the length of any representation of every
$h\in H_d$ is at most~$\eta(\acl d)$ for all $d\in D$ (see
also~\cite[Definition~3.3]{Ano21}). Of course, if $\fm H$ is
polynomially bounded, then it has exponential size. It
should be noted that a (weakly) pseudo-free family of
computational $\Omega$-algebras can have applications in
cryptography only if it is polynomially bounded or at least
has exponential size. Such families that do not have
exponential size \latin{per se} are of little interest; they
can be constructed unconditionally
(see~\cite[Subsection~3.4]{Ano21}). Finally, the family $\fm
H$ is said to have unique representations of elements if for
every $d\in D$, each element of $H_d$ is represented by a
unique bit string (see also~\cite[Definition~3.4]{Ano21}).
This property seems to be useful for applications. In this
paper, unless otherwise specified, families of computational
$\Omega$-algebras are assumed to be polynomially bounded and
to have unique representations of elements.

We emphasize that in the introduction, all results are
stated loosely. In particular, we do not mention the
probability distribution (depending on the security
parameter) according to which the index of the
$\Omega$-algebra in the family is sampled. Also, we usually
do not specify the representation of elements of the $\cl
V$-free $\Omega$-algebra by bit strings. This representation
is used for representing systems of the
form~\eqref{e:syseq}.

\subsection{Related Work}

Most researchers consider pseudo-freeness (in various
versions) in the varieties of all groups~\cite{Hoh03, Riv04,
Riv04pres, HT07, HIST09, Ano13}, of all abelian
groups~\cite{Hoh03, Riv04, Riv04pres, HT07, Mic10, JB09,
CFW11, FHIKS13, FHIS14ACISP, FHIS14IEICE, Ano18}, and of all
elementary abelian $p$-groups, where $p$ is a
prime~\cite{Ano17}. Surveys of this area can be found
in~\cite[Chapter~1]{Fuk14}, \cite[Section~1]{Ano18},
\cite[Subsection~1.1]{Ano21},
and~\cite[Subsection~1.1]{Ano22}.

We mention some conjectures and results concerning (weakly)
pseudo-free families of computational groups. In these
conjectures and results, families of computational groups
are presented in the form $((G_d,\pd G_d)\st d\in D)$, where
$D\subseteq\B^*$, $G_d$ is a group whose every element is
represented by a unique bit string of length polynomial in
the length of $d$, and $\pd G_d$ is a probability
distribution on $G_d$ ($d\in D$). Thus, these families are
polynomially bounded and have unique representations of
elements, as assumed above. Of course, the multiplication,
the inversion, and computing the identity element in $G_d$
should be performed efficiently when $d$ is given.
Furthermore, given $(d,1^k)$, one can efficiently generate
random elements of $G_d$ according to a probability
distribution that is statistically $2^{-k}$-close to~$\pd
G_d$. For a positive integer $n$, denote by $\Z_n$ the set
$\{0,\dots,n-1\}$ considered as a ring under addition and
multiplication modulo $n$ and by $\UZ_n$ the group of units
of~$\Z_n$. Also, let $\suz S_n$ and $\suz O_n$ be the
subgroups of squares in $\UZ_n$ (i.e., $\{z^2\bmod n\st
z\in\UZ_n\}$) and of elements of odd order in $\UZ_n$,
respectively. We denote by $\pd U(Y)$ the uniform
probability distribution on a nonempty finite set~$Y$.

Suppose $N$ is the set of all products of two distinct
primes. Rivest conjectured that the family $((\UZ_n,\pd
U(\UZ_n))\st n\in N)$ is pseudo-free in the variety $\cl A$
of all abelian groups (super-strong RSA conjecture,
see~\cite[Conjecture~1]{Riv04}, \cite[Slide~18]{Riv04pres}).
If both $p$ and $2p+1$ are prime numbers, then $p$ is called
a \emph{Sophie Germain prime} and $2p+1$ is said to be a
\emph{safe prime}. Let $S$ be the set of all products of two
distinct safe primes. Micciancio~\cite{Mic10} proved that
the family $((\UZ_n,\pd U(\suz S_n))\st n\in S)$ is
pseudo-free in $\cl A$ under the strong RSA assumption for
$S$ as the set of moduli. Informally, this assumption is
that, given a random $n\in S$ (for a given security
parameter) and a uniformly random $g\in\UZ_n$, it is
computationally hard to find an integer $e\ge2$ together
with an $e$th root of $g$ in~$\UZ_n$. It is easy to see that
if $n\in S$ and the prime factors of $n$ are different from
$5$, then $\suz S_n=\suz O_n$. Hence the above result of
Micciancio remains valid if we replace $\suz S_n$ by $\suz
O_n$ in it. The same result as in~\cite{Mic10}, but with
slightly different representations of group elements by bit
strings and different distributions of random elements of
the groups, was obtained by Jhanwar and Barua~\cite{JB09}.
Moreover, Catalano, Fiore, and Warinschi~\cite{CFW11} proved
that under the same assumption as in the above result of
Micciancio, the family $((\UZ_n,\pd U(\suz S_n))\st n\in S)$
satisfies an apparently stronger condition than
pseudo-freeness in~$\cl A$. That condition, called adaptive
pseudo-freeness, was introduced in~\cite{CFW11}.

Note that it is unknown whether the set $S$ is infinite.
Indeed, this holds if and only if there are infinitely many
Sophie Germain primes, which is a well-known unproven
conjecture in number theory. Thus, the assumption used
in~\cite{Mic10, JB09, CFW11} is very strong.

A natural candidate for a pseudo-free family in the variety
of all groups is $((\GL_2(\Z_n),\pd U(\GL_2(\Z_n)))\st n\in
N)$, where $\GL_2(\Z_n)$ is the group of invertible
$2\times2$ matrices over $\Z_n$ (see~\cite{CV13}). However,
(weak) pseudo-freeness of this family under a standard
cryptographic assumption is still unproven. Assume that
finding a nontrivial divisor of a random number in some set
$C$ of composite numbers (for a given security parameter) is
a computationally hard problem. Then Anokhin~\cite{Ano13}
constructed an exponential-size pseudo-free family in the
variety of all groups. That family is not polynomially
bounded and does not have unique representations of
elements. Moreover, each element of any group in that family
is represented by infinitely many bit strings. Note that the
family presented in~\cite{Ano13} is pseudo-free with respect
to a natural but non-succinct representation of elements of
the free group by bit strings. Under the same assumption,
Anokhin~\cite{Ano18} proved that the family $((\suz O_n,\pd
U(\suz O_n))\st n\in C)$ is weakly pseudo-free in~$\cl A$.
It is evident that this result also holds for $((\UZ_n,\pd
U(\suz O_n))\st n\in C)$. Compared to the above result of
Micciancio, this is a weaker statement, but it is proved
under a much weaker cryptographic assumption.

Suppose $p$ is an arbitrary fixed prime number and let $\cl
A_p$ be the variety of all elementary abelian $p$-groups.
Then pseudo-free families in $\cl A_p$ exist if and only if
certain homomorphic collision-resistant $p$-ary hash
function families exist or, equivalently, certain
homomorphic one-way families of functions exist.
See~\cite[Theorem~4.12]{Ano17} for details. Note that for
families of computational elementary abelian $p$-groups,
pseudo-freeness in $\cl A_p$ is equivalent to weak
pseudo-freeness in $\cl A_p$
(see~\cite[Theorem~3.7]{Ano17}).

There are many constructions of cryptographic objects based
on classical algebraic structures (e.g., groups). However,
to the best of our knowledge, there are only a few works
concerning both universal algebra and cryptography. Probably
the first such work is by Artamonov and
Yashchenko~\cite{AY94}. In that work, the authors introduced
and studied the notion of a pk-algebra. This notion
naturally formalizes the syntax of a two-message two-party
key agreement scheme. See also the extended
version~\cite{AKSY94} of~\cite{AY94}. Partala~\cite{Par18}
proposed a generalization of the well-known Diffie--Hellman
key agreement scheme based on universal algebras. Moreover,
he considered some approaches to the instantiation of the
proposed scheme. Loosely speaking, that scheme is secure if
it is computationally hard to compute images under an
unknown homomorphism (in a certain setting). See
also~\cite{Par11} (a preliminary version of~\cite{Par18})
and the thesis~\cite{Par15}.

Anokhin~\cite{Ano21} initiated the study of (weakly)
pseudo-free families of computational $\Omega$-algebras in
arbitrary varieties of $\Omega$-algebras. In our opinion,
the study of these families opens up new opportunities for
using (weak) pseudo-freeness in mathematical cryptography.
We briefly recall the main results of~\cite{Ano21}.

Let $\cl O$ denote the variety of all $\Omega$-algebras.
Then the following trichotomy holds:
\begin{roenum}
\item If $\Omega$ consists of nullary operation symbols
only, then unconditionally there exists a pseudo-free family
in~$\cl O$. This family consists of free $\Omega$-algebras.

\item If $\Omega=\Omega_0\cup\{\omega\}$, where $\Omega_0$
consists of nullary operation symbols and the arity of
$\omega$ is~$1$, then in $\cl O$, unconditionally there
exist an exponential-size pseudo-free family and a weakly
pseudo-free family. The former family has unique
representations of elements but is not polynomially bounded.

\item In all other cases, the existence of polynomially
bounded weakly pseudo-free families in $\cl O$ (not
necessarily having unique representations of elements)
implies the existence of collision-resistant hash function
families.
\end{roenum}

Assume that $\Omega$ contains a binary operation symbol
$\omega$ and $\cl V$ is a nontrivial variety of
$\Omega$-algebras such that any $\Omega$-algebra in $\cl V$
is a groupoid with an identity element under~$\omega$. (In
particular, this holds if $\cl V$ is a nontrivial variety of
monoids, loops, groups, or rings.) Then the existence of
polynomially bounded weakly pseudo-free families in $\cl V$
(not necessarily having unique representations of elements)
implies the existence of collision-resistant hash function
families. See~\cite[Section~4]{Ano21} for details.

Suppose $\Omega$ consists of a single $m$-ary operation
symbol, where $m\ge1$. In other words, we consider $m$-ary
groupoids. Furthermore, assume the existence of
collision-resistant hash function families. Then in $\cl O$,
there exist a weakly pseudo-free family and an
exponential-size pseudo-free family. The latter family is
not polynomially bounded and does not have unique
representations of elements. See~\cite[Section~5]{Ano21} for
details. As we have already seen, if $m=1$, then such
families (even an exponential-size pseudo-free family having
unique representations of elements) exist unconditionally.

In~\cite{Ano22}, Anokhin studied the connections between
pseudo-free families of computational $\Omega$-algebras (in
appropriate varieties of $\Omega$-algebras) and certain
standard cryptographic primitives. The main results of that
paper are as follows:
\begin{itemize}
\item Any $1$-pseudo-free (in particular, pseudo-free)
family of computational mono-unary algebras with one-to-one
fundamental operation (satisfying an additional condition)
in $\cl O$ naturally defines a one-way family of
permutations. Conversely, if there exists a one-way family
of permutations, then there exists a pseudo-free family of
computational mono-unary algebras in $\cl O$ with one-to-one
fundamental operation.

\item Let $m\in\{2,3,\dotsc\}$. Then any $1$-pseudo-free (in
particular, pseudo-free) family of computational $m$-unary
algebras with one-to-one fundamental operations (satisfying
an additional condition) in $\cl O$ naturally defines a
claw-resistant family of $m$-tuples of permutations.
Conversely, if there exists a claw-resistant family of
$m$-tuples of permutations, then there exists a pseudo-free
family of computational $m$-unary algebras in $\cl O$ with
one-to-one fundamental operations.

\item For a certain $\Omega$ and a certain variety $\cl V$
of $\Omega$-algebras, any $1$-pseudo-free (in particular,
pseudo-free) family of computational $\Omega$-algebras
(satisfying some additional conditions) in $\cl V$ naturally
defines a family of trapdoor permutations.
\end{itemize}
Recall that if $\Omega$ consists of a single unary operation
symbol (resp., of $m$ unary operation symbols), then
$\Omega$-algebras are called mono-unary (resp., $m$-unary)
algebras.

\subsection{Our Contribution and Organization of the Paper}

We note that the weak pseudo-freeness of all known
candidates for (weakly) pseudo-free families in nontrivial
varieties of groups can be broken (i.e., a
system~\eqref{e:syseq} of equations with $n=0$ can be found)
by efficient quantum algorithms. This raises the following
question: Does there exist (under a standard cryptographic
assumption) a post-quantum weakly pseudo-free family in some
nontrivial variety of groups? (The term ``post-quantum
weakly pseudo-free family'' means here that finding a
system~\eqref{e:syseq} with $n=0$ is computationally hard
even for quantum algorithms.) Recall that, unless otherwise
specified, families of computational $\Omega$-algebras are
assumed to be polynomially bounded and to have unique
representations of elements. Of course, all families of
computational $\Omega$-algebras (in particular, groups) in
the trivial variety are post-quantum pseudo-free in it. This
is because every system of equations of the
form~\eqref{e:syseq} is satisfiable in the $\cl V$-free
$\Omega$-algebra freely generated by $a_1,\dots,a_m$, which
is trivial if $m\ge1$. See~\cite[Remark~3.4]{Ano13} for
groups and~\cite[Remark~3.7]{Ano21} for $\Omega$-algebras.

In this paper, we also consider a worst-case version of weak
pseudo-freeness. In this version, loosely speaking, a member
$H$ of the family and elements $g_1,\dots,g_m\in H$ (see the
informal definition at the beginning of the paper) are
arbitrary rather than random. It is easy to see that weak
pseudo-freeness implies worst-case weak pseudo-freeness in
the same variety (see Remark~\ref{r:wpsfreeiswcwpsfree}).

Moreover, in addition to families of computational
$\Omega$-algebras, we consider families of black-box
$\Omega$-algebras. In the black-box $\Omega$-algebra model,
elements of a finite $\Omega$-algebra $H$ are represented
for computational purposes by bit strings of the same length
(depending on $H$) and the fundamental operations of $H$ are
performed by an oracle. This model was introduced by Babai
and Szemer\'edi~\cite{BS84} for groups. See
Subsection~\ref{ss:bbOamod} for details.

The above properties of families of computational
$\Omega$-algebras (pseudo-freeness, weak pseudo-freeness,
polynomial boundedness, etc.) can be defined for families of
black-box $\Omega$-algebras similarly. Like families of
computational $\Omega$-algebras, unless otherwise specified,
families of black-box $\Omega$-algebras are assumed to be
polynomially bounded and to have unique representations of
elements. Note that if there exists a weakly pseudo-free
family of computational $\Omega$-algebras, then there exists
a weakly pseudo-free family of black-box $\Omega$-algebras
in the same variety (see
Proposition~\ref{p:exwpffocOaimplexwpffobbOa}).

Let $\cl V$ be a variety of $\Omega$-algebras. For any
$\Psi\subseteq\Omega$, we denote by $\cl V|_\Psi$ the
variety of $\Psi$-algebras generated by the $\Psi$-reducts
of all $\Omega$-algebras in~$\cl V$. (The $\Psi$-reduct of
an $\Omega$-algebra $H$ is $H$ considered as a
$\Psi$-algebra.) Suppose $\Omega$ contains a set $\Gamma$ of
group operation symbols such that $\cl V|_\Gamma$ is a
variety of groups. (A set of group operation symbols
consists of a binary, a unary, and a nullary operation
symbols for the multiplication, the inversion, and the
identity element in a group, respectively.) In this case,
$\cl V$ is called a variety of expanded groups. Choose such
a set~$\Gamma$. Furthermore, we assume that $\cl V$ is
nontrivial and elements of the $\cl V$-free $\Omega$-algebra
freely generated by $a_1,a_2,\dotsc$ are represented by
straight-line programs (see Example~\ref{ex:SLPrepres}).
Then our main result (Theorem~\ref{t:tharenopqwpffoeg})
states that in $\cl V$, there are no families of any of the
following types:
\begin{itemize}
\item post-quantum weakly pseudo-free families of
computational $\Omega$-algebras,

\item post-quantum worst-case weakly pseudo-free families of
computational $\Omega$-algebras,

\item post-quantum weakly pseudo-free families of black-box
$\Omega$-algebras,

\item post-quantum worst-case weakly pseudo-free families of
black-box $\Omega$-algebras.
\end{itemize}
In particular, this is true for nontrivial varieties of
groups, rings, modules and algebras over a finitely
generated commutative associative ring with $1$, near-rings,
and, more generally, groups with finitely many multiple
operators (see Remark~\ref{r:apploftheorem}). Thus, we give
a negative answer to the above question. Note that if the
set $\Gamma$ cannot be chosen so that $\cl V|_\Gamma$ has
infinite exponent or is solvable, then our main result
depends on the classification of finite simple groups.

The outline of the proof of our main result is as follows.
First, it is sufficient to prove the nonexistence of
post-quantum worst-case weakly pseudo-free families of
black-box $\Omega$-algebras in $\cl V$ (see, e.g.,
Figure~\ref{f:relations}). Second, the results of
Subsection~\ref{ss:wpffofreducts} imply that for this it
suffices to prove the nonexistence of post-quantum
worst-case weakly pseudo-free families of black-box groups
in $\cl V|_\Gamma$. Third, for any family of black-box
groups in $\cl V|_\Gamma$, we construct a polynomial-time
black-box group quantum algorithm that breaks the worst-case
weak pseudo-freeness of this family. This algorithm is based
on a polynomial-time black-box group quantum algorithm for
one of the following problems:
\begin{roenum}
\item\label{i:finds} Given a black-box group $G\in\cl
V|_\Gamma$ and an element $g\in G$, find a positive integer
$s$ such that $g^s=1$ (if $\cl V|_\Gamma$ has infinite
exponent).

\item\label{i:findSLP} Given a black-box group $G\in\cl
V|_\Gamma$ and elements $g_1,\dots,g_m,h\in G$ such that $h$
is in the subgroup generated by $g_1,\dots,g_m$, find a
straight-line program computing $h$ from $g_1,\dots,g_m$ (if
$\cl V|_\Gamma$ has finite exponent).
\end{roenum}
Such algorithms for these problems do exist. Indeed,
problem~\ref{i:finds} can be solved in quantum polynomial
time by Shor's order-finding algorithm
(see~\cite[Section~5]{Sho97}, \cite[Subsection~5.3.1]{NC10},
or~\cite[Subsections~13.4--13.6]{KSV02}) modified for the
black-box group model. Problem~\ref{i:findSLP} can be solved
by the black-box group quantum algorithm of Ivanyos,
Magniez, and Santha for the constructive membership problem
(see~\cite[Theorem~5]{IMS03}). If $\cl V|_\Gamma$ is not the
variety of all groups (in particular, if $\cl V|_\Gamma$ has
finite exponent), then that algorithm runs in polynomial
time whenever the given black-box group is in~$\cl
V|_\Gamma$. This follows from a result of Jones~\cite{Jon74}
together with the classification of finite simple groups.
See Remark~\ref{r:ptqalgforCMP} for details.

For a positive integer $e$, $\cl A_e$ denotes the variety of
all abelian groups $G$ such that $g^e=1$ for all $g\in G$.
We note that if $\cl V|_\Gamma=\cl A_e$, where $e\ge2$, then
the third step of the proof of our main result can be done
using a polynomial-time quantum algorithm for the hidden
subgroup problem for the $\cl A_e$-free group generated by
$a_1,\dots,a_m$ ($m\ge1$). (This group is the direct product
of the cyclic subgroups generated by $a_1,\dots,a_m$; each
of these subgroups has order~$e$.) Such an algorithm exists,
e.g., by~\cite[Theorem~3.13]{Lom04}. We
recommend~\cite{Lom04} as a good source of information on
quantum algorithms for the hidden subgroup problem.

The rest of the paper is organized as follows.
Section~\ref{s:prelim} contains notation, basic definitions,
and general results used in the paper. In
Section~\ref{s:wpsfreefams}, we define and discuss some
types of weak pseudo-freeness (including post-quantum ones)
for families of computational and black-box
$\Omega$-algebras. Relations between these types are studied
in Subsection~\ref{ss:relations} and depicted in
Figure~\ref{f:relations}. In our opinion, some of these
types of weak pseudo-freeness might be interesting for
future research. Also, we want to state our main result in
the strongest possible form. This is another motivation for
introducing new types of weak pseudo-freeness. In
Subsection~\ref{ss:wpffofreducts}, loosely speaking, we show
that if $\Psi\subseteq\Omega$, then the family of
$\Psi$-reducts of $\Omega$-algebras in a weakly pseudo-free
family in $\cl V$ is weakly pseudo-free in~$\cl V|_\Psi$.
The purpose of Section~\ref{s:ptbbgrqalgs} is to prove the
existence of polynomial-time black-box group quantum
algorithms that are used in the proof of our main result.
These algorithms are constructed in the proofs of
Lemmas~\ref{l:ProutBinLambdainfexp}
and~\ref{l:ProutBinLambdantrnall}. In
Section~\ref{s:mainres}, we prove the main result of this
paper (Theorem~\ref{t:tharenopqwpffoeg}).
Section~\ref{s:conclanddirforfutres} concludes and suggests
some directions for future research. Finally, in
Appendix~\ref{a:tabofnot}, we briefly recall some notation
introduced in Sections~\ref{s:prelim}
and~\ref{s:wpsfreefams}.

\section{Preliminaries}
\label{s:prelim}

\subsection{General Preliminaries}

In this paper, $\N$ denotes the set of all nonnegative
integers. Let $Y$ be a set and let $n\in\N$. We denote by
$Y^n$ the set of all (ordered) $n$-tuples of elements
from~$Y$. Of course, $Y^1$ is identified with~$Y$.
Furthermore, we put $Y^{\le n}=\bigcup_{i=0}^nY^i$ and
$Y^*=\bigcup_{i=0}^\infty Y^i$. In particular, $\emptyset^*$
consists only of the empty tuple.

We consider elements of $\B^*$ as bit strings and denote the
length of a string $u\in\B^*$ by~$\acl u$. The unary
representation of $n$, i.e., the string of $n$ ones, is
denoted by~$1^n$. Similarly, $0^n$ is the string of $n$
zeros. As usual, $\oplus$ denotes the bitwise XOR operation.

Let $I$ be a set. Suppose each $i\in I$ is assigned an
object~$q_i$. Then we denote by $(q_i\st i\in I)$ the family
of all these objects, whereas $\{q_i\st i\in I\}$ denotes
the set of all elements of this family.

When necessary, we assume that all ``finite'' objects (e.g.,
integers, tuples of integers, tuples of tuples of integers)
are represented by bit strings in some natural way.
Sometimes we identify such objects with their
representations. Unless otherwise specified, integers are
represented by their binary expansions.

Suppose $\phi$ is a function. We denote by $\dom\phi$ the
domain of~$\phi$. Also, we use the same notation for $\phi$
and for the function
$(z_1,\dots,z_n)\mapsto(\phi(z_1),\dots,\phi(z_n))$, where
$n\in\N$ and $z_1,\dots,z_n\in\dom\phi$. The identity
function on the set $Y$ is denoted by~$\id_Y$.

Let $\rho$ be a function from a subset of $\B^*$ onto a set
$T$ and let $t\in T$. Then $\rep t\rho$ denotes an arbitrary
preimage of $t$ under~$\rho$. We use $\rep t\rho$ as a
representation of $t$ for computational purposes. A similar
notation was used by Boneh and Lipton in~\cite{BL96} and by
Hohenberger in~\cite{Hoh03}.

For convenience, we say that a function
$\pi\colon\N\to\N\setminus\{0\}$ is a \emph{polynomial} if
there exist $c\in\N\setminus\{0\}$ and $d\in\N$ such that
$\pi(n)=cn^d$ for any $n\in\N\setminus\{0\}$ ($\pi(0)$ can
be an arbitrary positive integer). Of course, every
polynomial growth function from $\N$ to $\R_+=\{r\in\R\st
r\ge0\}$ can be upper bounded by a polynomial in this sense.
Therefore this notion of a polynomial is sufficient for our
purposes.

\subsection{Universal-Algebraic Preliminaries}
\label{ss:univalgprelim}

In this subsection, we recall the basic definitions and
simple facts from universal algebra. For a detailed
introduction to this topic, the reader is referred to
standard books, e.g.,~\cite{Cohn81, BS12, Wech92}.

Throughout the paper, $\Omega$ denotes a set of finitary
operation symbols. Moreover, in all sections except this
one, we assume that $\Omega$ is finite. Each
$\omega\in\Omega$ is assigned a nonnegative integer called
the \emph{arity} of $\omega$ and denoted by~$\ar\omega$. An
\emph{$\Omega$-algebra} is a set $H$ called the
\emph{carrier} (or the \emph{underlying set}) together with
a family $(\widehat\omega\colon H^{\ar\omega}\to
H\st\omega\in\Omega)$ of operations on $H$ called the
\emph{fundamental operations}. For simplicity of notation,
the fundamental operation $\widehat\omega$ associated with a
symbol $\omega\in\Omega$ will be denoted by~$\omega$.
Furthermore, we often denote an $\Omega$-algebra and its
carrier by the same symbol.

Let $H$ be an $\Omega$-algebra. A subset of $H$ is called a
\emph{subalgebra} of $H$ if it is closed under the
fundamental operations of~$H$. If $S$ is a system of
elements of $H$, then we denote by $\alg S$ the subalgebra
of $H$ generated by $S$, i.e., the smallest subalgebra of
$H$ containing~$S$.

Suppose $G$ is an $\Omega$-algebra. A \emph{homomorphism} of
$G$ to $H$ is a function $\phi\colon G\to H$ such that for
every $\omega\in\Omega$ and $g_1,\dots,g_{\ar\omega}\in G$,
\[
\phi(\omega(g_1,\dots,g_{\ar\omega}))
=\omega(\phi(g_1),\dots,\phi(g_{\ar\omega})).
\]
If a homomorphism of $G$ onto $H$ is one-to-one, then it is
called an \emph{isomorphism}. Of course, the
$\Omega$-algebras $G$ and $H$ are said to be
\emph{isomorphic} if there exists an isomorphism of $G$
onto~$H$.

Let $(H_i\st i\in I)$ be a family of $\Omega$-algebras. The
fundamental operations of the \emph{direct product} of this
family are defined as follows:
\[
\omega((h_{1,i}\st i\in I),\dots,(h_{\ar\omega,i}\st i\in
I))=(\omega(h_{1,i},\dots,h_{\ar\omega,i})\st i\in I),
\]
where $\omega\in\Omega$ and
$h_{1,i},\dots,h_{\ar\omega,i}\in H_i$ for all $i\in I$.

An $\Omega$-algebra with only one element is said to be
\emph{trivial}. It is obvious that all trivial
$\Omega$-algebras are isomorphic.

For every $n\in\N$, put
$\Omega_n=\{\omega\in\Omega\st\ar\omega=n\}$. We note that
if $\Omega_0=\emptyset$, then an $\Omega$-algebra may be
empty. Whenever $\omega\in\Omega_0$, it is common to write
$\omega$ instead of~$\omega()$.

Let $Z$ be a set of objects called variables. We always
assume that any variable is not in~$\Omega$. The set $\Tm Z$
of all \emph{$\Omega$-terms} (or simply \emph{terms}) over
$Z$ is defined as the smallest set such that $\Omega_0\cup
Z\subseteq\Tm Z$ and if $\omega\in\Omega\setminus\Omega_0$
and $v_1,\dots,v_{\ar\omega}\in\Tm Z$, then the formal
expression $\omega(v_1,\dots,v_{\ar\omega})$ is in~$\Tm Z$.
Of course, $\Tm Z$ is an $\Omega$-algebra under the natural
fundamental operations. This $\Omega$-algebra is called the
\emph{$\Omega$-term algebra} over~$Z$.

Consider the case when $Z=\{z_1,z_2,\dotsc\}$, where
$z_1,z_2,\dotsc$ are distinct. Let $m\in\N$. We denote by
$T_\infty$ and $T_m$ the $\Omega$-term algebras
$\Tm\{z_1,z_2,\dotsc\}$ and $\Tm\{z_1,\dots,z_m\}$,
respectively. Suppose $h=(h_1,\dots,h_m,\dotsc)$ is either
an $m'$-tuple, where $m'\ge m$, or an infinite sequence of
elements of~$H$. Then for every $v\in T_m$, the element
$v(h)=v(h_1,\dots,h_m)\in H$ is defined inductively in the
natural way. It is easy to see that $\{v(h_1,\dots,h_m)\st
v\in T_m\}=\alg{h_1,\dots,h_m}$. If $h$ is an infinite
sequence, then $\{v(h)\st v\in
T_\infty\}=\alg{h_1,h_2,\dotsc}$.

An \emph{identity} (or a \emph{law}) over $\Omega$ is a
closed first-order formula of the form
$\Forall{z_1,\dots,z_m}(v=w)$, where $m\in\N$ and $v,w\in
T_m$. Usually we will omit the phrase ``over $\Omega$.'' We
will write identities simply as $v=w$, where $v,w\in
T_\infty$, assuming that all variables are universally
quantified. A class $\cl V$ of $\Omega$-algebras is said to
be a \emph{variety} if it can be defined by a set $\Upsilon$
of identities. This means that for any $\Omega$-algebra $G$,
$G\in\cl V$ if and only if $G$ satisfies all identities
in~$\Upsilon$. By the famous Birkhoff variety theorem (see,
e.g.,~\cite[Chapter~IV, Theorem~3.1]{Cohn81},
\cite[Chapter~II, Theorem~11.9]{BS12},
or~\cite[Subsection~3.2.3, Theorem~21]{Wech92}), a class of
$\Omega$-algebras is a variety if and only if it is closed
under taking subalgebras, homomorphic images, and direct
products. Note that if a class of $\Omega$-algebras is
closed under taking direct products, then it contains a
trivial $\Omega$-algebra as the direct product of the empty
family of $\Omega$-algebras.

The variety consisting of all $\Omega$-algebras with at most
one element is said to be \emph{trivial}; all other
varieties of $\Omega$-algebras are called \emph{nontrivial}.
The trivial variety is defined by the identity $z_1=z_2$.
When $\Omega_0=\emptyset$, the trivial variety contains not
only trivial $\Omega$-algebras, but also the empty
$\Omega$-algebra. If $\cl C$ is a class of
$\Omega$-algebras, then the variety \emph{generated} by $\cl
C$ is the smallest variety of $\Omega$-algebras
containing~$\cl C$. This variety is defined by the set of
all identities holding in all $\Omega$-algebras in~$\cl C$.

Throughout the paper, $\cl V$ denotes a variety of
$\Omega$-algebras. An $\Omega$-algebra $F\in\cl V$ is said
to be \emph{$\cl V$-free} if it has a generating system
$(f_i\st i\in I)$ such that for every system of elements
$(g_i\st i\in I)$ of any $\Omega$-algebra $G\in\cl V$ there
exists a homomorphism $\alpha\colon F\to G$ satisfying
$\alpha(f_i)=g_i$ for all $i\in I$ (evidently, this
homomorphism $\alpha$ is unique). Any generating system
$(f_i\st i\in I)$ with this property is called \emph{free}
and the $\Omega$-algebra $F$ is said to be \emph{freely
generated} by every such system. The next lemma is well
known and can be proved straightforwardly.

\begin{lemma}
\label{l:freegensys}
Suppose $F$ is an $\Omega$-algebra in $\cl V$ and $(f_i\st
i\in I)$ is a generating system of~$F$. Then $F$ is a $\cl
V$-free $\Omega$-algebra freely generated by $(f_i\st i\in
I)$ if and only if for any $m\in\N$ and any $v,w\in T_m$,
the identity $v=w$ holds in $\cl V$ whenever
$v(f_{i_1},\dots,f_{i_m})=w(f_{i_1},\dots,f_{i_m})$ for some
distinct $i_1,\dots,i_m\in I$.
\end{lemma}

It is well known (see, e.g.,~\cite[Chapter~IV,
Corollary~3.3]{Cohn81}, \cite[Chapter~II, Definition~10.9
and Theorem~10.10]{BS12}, or~\cite[Subsection~3.2.3,
Theorem~16]{Wech92}) that for any set $I$ there exists a
unique (up to isomorphism) $\cl V$-free $\Omega$-algebra
with a free generating system indexed by~$I$. It is easy to
see that if $\cl V$ is nontrivial, then for every free
generating system $(f_i\st i\in I)$ of a $\cl V$-free
$\Omega$-algebra, $f_i$ are distinct. In this case, one may
consider free generating systems as sets.

We denote by $F_\infty(\cl V)$ the $\cl V$-free
$\Omega$-algebra freely generated by $a_1,a_2,\dotsc$. Of
course, if $\cl V$ is nontrivial, then $a_1,a_2,\dotsc$ are
assumed to be distinct. Furthermore, suppose $m\in\N$ and
let $F_m(\cl V)=\alg{a_1,\dots,a_m}$. For elements of
$F_m(\cl V)$, we use the notation $v(a)=v(a_1,\dots,a_m)$,
where $v\in T_m$. It is well known that $a_1,a_2,\dotsc$ can
be considered as variables taking values in arbitrary
$\Omega$-algebra $G\in\cl V$. That is, for any $v(a)\in
F_m(\cl V)$ and any $g=(g_1,\dots,g_m)\in G^m$, the element
$v(g)=v(g_1,\dots,g_m)\in G$ is well defined as
$\alpha(v(a))$, where $\alpha$ is the unique homomorphism of
$F_m(\cl V)$ to $G$ such that $\alpha(a_i)=g_i$ for all
$i\in\{1,\dots,m\}$.

By a \emph{straight-line program} over $\Omega$ we mean a
nonempty sequence $(u_1,\dots,u_n)$ such that for every
$i\in\{1,\dots,n\}$, either $u_i\in\N\setminus\{0\}$ or
$u_i=(\omega,j_1,\dots,j_{\ar\omega})$, where
$\omega\in\Omega$ and
$j_1,\dots,j_{\ar\omega}\in\{1,\dots,i-1\}$. These two cases
should be clearly distinguished. Usually we will omit the
phrase ``over $\Omega$.'' Suppose $g_1,\dots,g_m\in H$,
where $m\in\N$. Furthermore, let $u=(u_1,\dots,u_n)$ be a
straight-line program such that if $u_i\in\N\setminus\{0\}$,
then $u_i\le m$. Then $u$ naturally defines the sequence
$(h_1,\dots,h_n)$ of elements of $H$ by induction. Namely,
for each $i\in\{1,\dots,n\}$, we put $h_i=g_{u_i}$ if
$u_i\in\N\setminus\{0\}$ and
$h_i=\omega(h_{j_1},\dots,h_{j_{\ar\omega}})$ if
$u_i=(\omega,j_1,\dots,j_{\ar\omega})$, where $\omega$ and
$j_1,\dots,j_{\ar\omega}$ are as above. We say that $u$
\emph{computes} the element $h_n$ from $g_1,\dots,g_m$. The
positive integer $n$ is called the \emph{length} of the
straight-line program~$u$. It is easy to see that an element
$h\in H$ can be computed from $g_1,\dots,g_m$ by a
straight-line program if and only if
$h\in\alg{g_1,\dots,g_m}$.

\subsection{Group-Theoretic Preliminaries}

In this subsection, we recall some definitions and facts
from group theory. For a detailed introduction to this
topic, the reader is referred to standard textbooks,
e.g.,~\cite{KM79, Rob96, Rot95}.

We say that $\Omega$ is a \emph{set of group operation
symbols} if it consists of a binary, a unary, and a nullary
operation symbols (for the multiplication, the inversion,
and the identity element in a group, respectively). We
consider groups as $\Omega$-algebras, where $\Omega$ is a
set of group operation symbols. Therefore the content of
Subsection~\ref{ss:univalgprelim} is applicable to groups.
Of course, we use the standard group-theoretic notation,
e.g., $gh$, $g^n$, and $1$, where $g$ and $h$ are elements
of a group and $n$ is an integer.

The abbreviation CFSG stands for the Classification of
Finite Simple Groups. This classification states that every
finite simple group is isomorphic to
\begin{itemize}
\item a cyclic group of prime order,

\item an alternating group of degree at least~$5$,

\item a finite simple group of Lie type, or

\item one of the $26$~sporadic finite simple groups.
\end{itemize}
See~\cite[Chapter~2]{Gor82} or~\cite[Part~I, Chapter~1,
Section~1]{GLS94} for details.

Let $G$ be a group. The notation $H\pns G$ means that $H$ is
a proper normal subgroup of~$G$. A subnormal series
\[
\{1\}=G_0\pns G_1\pns\dots\pns G_n=G
\]
is said to be a \emph{composition series} of the group $G$
if all factors $G_i/G_{i-1}$ ($i\in\{1,\dots,n\}$) of this
series are simple groups. Of course, not every group has a
composition series. However, any finite group certainly has
one. By the well-known Jordan--H\"older theorem, the factors
of a composition series of $G$ do not depend on the series
(up to isomorphism and permutation of factors); these
factors are called the \emph{composition factors} of~$G$.
See~\cite[Chapter~5, Section ``The Jordan--H\"older
Theorem'']{Rot95}, \cite[Section~3.1]{Rob96},
or~\cite[Subsection~4.4]{KM79}. For finite groups, see also
\cite[Section~1.1, Definition~D6]{Gor82} or~\cite[Part~I,
Chapter~1, Section~3]{GLS94}. It is well known and easy to
see that a finite group is solvable if and only if all its
composition factors are abelian (or, equivalently, cyclic of
prime order).

The next lemma, due to Babai and Szemer\'edi, is known as
the Reachability Theorem or the Reachability Lemma
(see~\cite[Theorem~3.1]{BS84} or~\cite[Lemma~6.4]{Bab91}).

\begin{lemma}
\label{l:shortSLP}
Suppose $G$ is a finite group. Let $g_1,\dots,g_m$ (where
$m\in\N$) be a generating system of~$G$. Then any element of
$G$ can be computed from $g_1,\dots,g_m$ by a straight-line
program of length at most $(1+\log_2\acl G)^2$.
\end{lemma}

Suppose $\cl W$ is a variety of groups. Assume that there
exists a positive integer $n$ such that every group in $\cl
W$ satisfies the identity $z_1^n=1$. Then the smallest such
positive integer is called the \emph{exponent} of the
variety~$\cl W$. Otherwise the \emph{exponent} of $\cl W$ is
said to be infinite. In the latter case, some authors say
that $\cl W$ is of exponent zero (see, e.g.,~\cite{Neu67}).
It is easy to see that the exponent of $\cl W$ (finite or
infinite) coincides with~$\acl{F_1(\cl W)}$.

The variety $\cl W$ is called \emph{solvable} if it consists
of solvable groups. It is evident that the derived length of
groups in any solvable variety is upper bounded by a
nonnegative integer depending on the variety.

\subsection{Probabilistic Preliminaries}

Let $\pd Y$ be a probability distribution on a finite or
countably infinite sample space~$Y$. Then we denote by
$\supp\pd Y$ the \emph{support} of $\pd Y$, i.e., the set
$\{y\in Y\st\Pro_{\pd Y}\{y\}\ne0\}$. In many cases, one can
consider $\pd Y$ as a distribution on~$\supp\pd Y$.

Suppose $Z$ is a finite or countably infinite set and
$\alpha$ is a function from $Y$ to~$Z$. Then the image of
$\pd Y$ under $\alpha$, which is a probability distribution
on $Z$, is denoted by $\alpha(\pd Y)$. This distribution is
defined by $\Pro_{\alpha(\pd Y)}\{z\}=\Pro_{\pd
Y}\alpha^{-1}(z)$ for each $z\in Z$. Note that if a random
variable $\rv y$ is distributed according to $\pd Y$, then
the random variable $\alpha(\rv y)$ is distributed according
to~$\alpha(\pd Y)$.

We use the notation $\rv y_1,\dots,\rv y_n\rvd\pd Y$ to
indicate that $\rv y_1,\dots,\rv y_n$ (denoted by upright
bold letters) are independent random variables distributed
according to~$\pd Y$. We assume that these random variables
are independent of all other random variables defined in
such a way. Furthermore, all occurrences of an upright bold
letter in a probabilistic statement refer to the same
(unique) random variable. Of course, all random variables in
a probabilistic statement are assumed to be defined on the
same sample space. Other specifics of random variables do
not matter for us. Note that the probability distribution
$\pd Y$ in this notation may be random. For example, let
$(\pd Y_i\st i\in I)$ be a probability ensemble consisting
of distributions on the set $Y$, where the set $I$ is finite
or countably infinite. Moreover, suppose $\pd I$ is a
probability distribution on~$I$. Then $\rv i\rvd\pd I$ and
$\rv y\rvd\pd Y_{\rv i}$ mean that the joint distribution of
the random variables $\rv i$ and $\rv y$ is given by
$\Prob{\rv i=i\cs\rv y=y}=\Pro_{\pd I}\{i\}\Pro_{\pd
Y_i}\{y\}$ for each $i\in I$ and $y\in Y$.

For any $n\in\N$, we denote by $\pd Y^n$ the distribution of
a random variable $(\rv y_1,\dots,\rv y_n)$, where $\rv
y_1,\dots,\rv y_n\rvd\pd Y$. (Of course, the distribution of
this random variable does not depend on the choice of
independent random variables $\rv y_1,\dots,\rv y_n$
distributed according to~$\pd Y$.) It is easy to see that
$\alpha(\pd Y^n)=(\alpha(\pd Y))^n$ for every $\alpha\colon
Y\to Z$ and $n\in\N$.

\subsection{Cryptographic Preliminaries}

Let $\pd P=(\pd P_i\st i\in I)$ be a probability ensemble
consisting of distributions on~$\B^*$, where
$I\subseteq\B^*$. Then $\pd P$ is called
\emph{polynomial-time samplable} (or \emph{polynomial-time
constructible}) if there exists a probabilistic
polynomial-time algorithm $A$ such that for every $i\in I$
the random variable $A(i)$ is distributed according to~$\pd
P_i$. It is easy to see that if $\pd P$ is polynomial-time
samplable, then there exists a polynomial $\pi$ satisfying
$\supp\pd P_i\subseteq\B^{\le\pi(\acl i)}$ for any $i\in I$.
Furthermore, let $\pd Q=(\pd Q_j\st j\in J)$ be a
probability ensemble consisting of distributions on~$\B^*$,
where $J\subseteq\N$. Usually, when it comes to
polynomial-time samplability of $\pd Q$, the indices are
assumed to be represented in binary. If, however, these
indices are represented in unary, then we specify this
explicitly. Thus, the ensemble $\pd Q$ is said to be
\emph{polynomial-time samplable when the indices are
represented in unary} if there exists a probabilistic
polynomial-time algorithm $B$ such that for every $j\in J$
the random variable $B(1^j)$ is distributed according
to~$\pd Q_j$.

Suppose $K$ is an infinite subset of $\N$ and $D$ is a
subset of~$\B^*$. Also, let $(\pd D_k\st k\in K)$ be a
probability ensemble consisting of distributions on~$D$. We
assume that this probability ensemble is polynomial-time
samplable when the indices are represented in unary.
Furthermore, suppose $(D_k\st k\in K)$ is a family of
nonempty subsets of $D$ such that there exists a polynomial
$\theta$ satisfying $D_k\subseteq\B^{\le\theta(k)}$ for all
$k\in K$. This notation is used throughout the paper.

A function $\delta\colon K\to\R_+$ is called
\emph{negligible} if for every polynomial $\pi$ there exists
a nonnegative integer $n$ such that $\delta(k)\le1/\pi(k)$
whenever $k\in K$ and $k\ge n$. We denote by $\negl$ an
unspecified negligible function on~$K$. Any equality
containing $\negl(k)$ is meant to hold for all $k\in K$.

\section{Weakly Pseudo-Free Families of Computational and
Black-Box \sectionheadOmega-Algebras}
\label{s:wpsfreefams}

From now on, we assume that $\Omega$ is finite. This allows
us to avoid representation issues. In this section, we
formally define and discuss families of computational and
black-box $\Omega$-algebras, as well as some types of weak
pseudo-freeness (including post-quantum ones) for these
families. Of course, one can easily define the respective
types of pseudo-freeness.

Throughout the paper, we denote by $\sigma$ a function from
a subset of $\B^*$ onto~$F_\infty(\cl V)$. This function is
used for representation of elements of $F_\infty(\cl V)$ for
computational purposes. Let $H\in\cl V$ and
$g=(g_1,\dots,g_m)$, where $m\in\N\setminus\{0\}$ and
$g_1,\dots,g_m\in H$. Then we put
\begin{align*}
\Lambda(H,\cl V,\sigma,g)&=\{(t,u)\in(\dom\sigma)^2\st
\sigma(t),\sigma(u)\in F_m(\cl
V)\cs\sigma(t)\ne\sigma(u)\cs\sigma(t)(g)=\sigma(u)(g)\}\\
&=\bigcup_{\substack{v,w\in F_m(\cl V)\text{ s.t.}\\v\ne
w\land v(g)=w(g)}}(\sigma^{-1}(v)\times\sigma^{-1}(w)).
\end{align*}
It is natural to call a pair $(v,w)\in(F_m(\cl V))^2$ a
\emph{nontrivial relation} between $g_1,\dots,g_m$ if $v\ne
w$ and $v(g)=w(g)$. Then $\Lambda(H,\cl V,\sigma,g)$ is the
set of all representations of nontrivial relations between
$g_1,\dots,g_m$ using~$\sigma$.

\begin{example}[{representation of elements of $F_\infty(\cl
V)$ by straight-line programs, see
also~\cite[Example~3.13]{Ano21}
or~\cite[Example~2.10]{Ano22}}]
\label{ex:SLPrepres}
Denote by $\SLP_{\cl V}$ the function that takes each
straight-line program $u$ (over $\Omega$) to the element of
$F_\infty(\cl V)$ computed by $u$ from $a_1,\dots,a_m$,
where the nonnegative integer $m$ is an upper bound for all
integer elements of the sequence~$u$. (Of course, this
element of $F_\infty(\cl V)$ does not depend on~$m$.)
Usually we will write $\SLP$ instead of~$\SLP_{\cl V}$. It
is evident that $\SLP$ is a function onto~$F_\infty(\cl V)$.
In most of our results, we will use this function as the
function~$\sigma$. Note that this method of representation
(for elements of the free group) was used in~\cite{Hoh03}.
\end{example}

\subsection{Standard Model}
\label{ss:standmod}

A general definition of a family of computational
$\Omega$-algebras was given in~\cite{Ano21} (see
Definition~3.1 in that work). These families consist of
triples of the form $(H_d,\rho_d,\pd R_d)$, where $d$ ranges
over $D$, $H_d$ is an $\Omega$-algebra, $\rho_d$ is a
function from a subset of $\B^*$ onto $H_d$, and $\pd R_d$
is a probability distribution on $\dom\rho_d$ for any $d\in
D$. In this paper, we study only polynomially bounded
families $((H_d,\rho_d,\pd R_d)\st d\in D)$ of computational
$\Omega$-algebras that have unique representations of
elements. This means that the following conditions hold:
\begin{itemize}
\item There exists a polynomial $\eta$ such that
$\dom\rho_d\subseteq\B^{\le\eta(\acl d)}$ for all $d\in D$.
See also~\cite[Definition~3.3]{Ano21}.

\item For each $d\in D$, the function $\rho_d$ is
one-to-one. Hence we can assume that for every $d\in D$,
$H_d\subseteq\B^*$ and the unique representation of each
$h\in H_d$ is $h$ itself. Namely, we use the family
$((\dom\rho_d,\id_{\dom\rho_d},\pd R_d)\st d\in D)$ instead
of $((H_d,\rho_d,\pd R_d)\st d\in D)$. Here $\dom\rho_d$ is
considered as the unique $\Omega$-algebra such that $\rho_d$
is an isomorphism of this $\Omega$-algebra onto $H_d$ ($d\in
D$). See also~\cite[Definition~3.4 and Remark~3.5]{Ano21}.
Moreover, if $H_d\subseteq\B^*$, then we write $(H_d,\pd
R_d)$ instead of $(H_d,\id_{H_d},\pd R_d)$.
\end{itemize}

Now we give a formal definition of a family of computational
$\Omega$-algebras with the above restrictions. We also need
a variant of this definition without probability
distributions.

Suppose an $\Omega$-algebra $H_d\subseteq\B^*$ is assigned
to each $d\in D$. When necessary, we denote by $\pd H_d$ a
probability distribution on the (necessarily nonempty)
$\Omega$-algebra $H_d$ for every $d\in D$. Note that some
definitions in this subsection do not depend on these
probability distributions.

\begin{definition}[family of computational $\Omega$-algebras
without distributions]
\label{d:focOawod}
The family $(H_d\st d\in D)$ is called a \emph{family of
computational $\Omega$-algebras without distributions} if
the following two conditions hold:
\begin{roenum}
\item There exists a polynomial $\eta$ such that
$H_d\subseteq\B^{\le\eta(\acl d)}$ for all $d\in D$.

\item For every $\omega\in\Omega$ there exists a
deterministic polynomial-time algorithm that, given $d\in D$
and $h_1,\dots,h_{\ar\omega}\in H_d$, computes
$\omega(h_1,\dots,h_{\ar\omega})$ in~$H_d$.
\end{roenum}
\end{definition}

\begin{definition}[{family of computational
$\Omega$-algebras (with distributions), see
also~\cite[Definition~3.1]{Ano21}
or~\cite[Definition~2.6]{Ano22}}]
\label{d:focOa}
The family $((H_d,\pd H_d)\st d\in D)$ is said to be a
\emph{family of computational $\Omega$-algebras with
distributions} or simply a \emph{family of computational
$\Omega$-algebras} if the following two conditions hold:
\begin{roenum}
\item The family $(H_d\st d\in D)$ is a family of
computational $\Omega$-algebras without distributions.

\item The probability ensemble $(\pd H_d\st d\in D)$ is
polynomial-time samplable.
\end{roenum}
\end{definition}

Thus, by default, a family of computational
$\Omega$-algebras is a family with distributions. The main
motivation for introducing the notion of a family of
computational $\Omega$-algebras without distributions is to
abstract from these distributions whenever this is possible.

\begin{definition}[family is in~$\cl V$]
We say that the family $(H_d\st d\in D)$ (or $((H_d,\pd
H_d)\st d\in D)$) is in $\cl V$ if $H_d\in\cl V$ for all
$d\in D$.
\end{definition}

\begin{definition}[weakly pseudo-free family of
computational $\Omega$-algebras]
\label{d:wpsfreefocOa}
Assume that $((H_d,\pd H_d)\st d\in D)$ is a family of
computational $\Omega$-algebras in~$\cl V$. Then this family
is called \emph{weakly pseudo-free} in $\cl V$ with respect
to $(\pd D_k\st k\in K)$ and $\sigma$ if for any polynomial
$\pi$ and any probabilistic polynomial-time algorithm~$A$,
\[
\Prob{A(1^k,\rv d,\rv g)\in\Lambda(H_{\rv d},\cl
V,\sigma,\rv g)}=\negl(k),
\]
where $\rv d\rvd\pd D_k$ and $\rv g\rvd\pd H_{\rv
d}^{\pi(k)}$.
\end{definition}

\begin{remark}
Note that for any $H\in\cl V$ and any $g\in H^m$ (where
$m\in\N\setminus\{0\}$), $\Lambda(H,\cl V,\sigma,g)$
coincides with $\Sigma_1'(H,\cl V,\sigma,g)$ in the notation
of~\cite{Ano21}. So weak pseudo-freeness in the sense of
Definition~\ref{d:wpsfreefocOa} is in fact weak
$1$-pseudo-freeness in the sense
of~\cite[Remark~3.9]{Ano21}. However, it is easy to see that
weak $1$-pseudo-freeness in $\cl V$ with respect to $(\pd
D_k\st k\in K)$ and $\sigma$ is equivalent to weak
pseudo-freeness in $\cl V$ with respect to $(\pd D_k\st k\in
K)$ and $\sigma$ (see~\cite[Remark~3.9]{Ano21}).
\end{remark}

\begin{definition}[worst-case weakly pseudo-free family of
computational $\Omega$-algebras without distributions]
\label{d:wcwpsfreefocOa}
Assume that $(H_d\st d\in D)$ is a family of nonempty
computational $\Omega$-algebras in $\cl V$ without
distributions. Then this family is said to be
\emph{worst-case weakly pseudo-free} in $\cl V$ with respect
to $(D_k\st k\in K)$ and $\sigma$ if for any polynomial
$\pi$ and any probabilistic polynomial-time algorithm~$A$,
\[
\min_{d\in D_k\cs g\in
H_d^{\pi(k)}}\Prob{A(1^k,d,g)\in\Lambda(H_d,\cl
V,\sigma,g)}=\negl(k).
\]
\end{definition}

\subsection{Black-Box \headOmega-Algebra Model}
\label{ss:bbOamod}

Babai and Szemer\'edi~\cite{BS84} introduced a model of
computation in finite groups, called the black-box group
model. In this model, elements of a finite group $G$ are
represented for computational purposes by bit strings of the
same length (depending on $G$) and the group operations in
$G$ are performed by an oracle. Such groups are called
black-box groups. This model can be naturally generalized to
$\Omega$-algebras.

In this paper, unless otherwise specified, we require every
element of a black-box $\Omega$-algebra to be represented by
a unique bit string. Therefore we can assume that for any
black-box $\Omega$-algebra $H$, we have $H\subseteq\B^n$,
where $n\in\N$, and the unique representation of each $h\in
H$ is $h$ itself.

\begin{definition}[black-box $\Omega$-algebra]
Any $\Omega$-algebra $H$ such that $H\subseteq\B^n$ for some
$n\in\N$ is called a \emph{black-box $\Omega$-algebra}.
\end{definition}

Let $H$ be a black-box $\Omega$-algebra. It is evident that
if $H\ne\emptyset$, then $H\subseteq\B^n$ for a single
$n\in\N$. Otherwise this inclusion holds for all $n\in\N$.

\begin{definition}[$\Omega$-oracle]
\label{d:Omegaoracle}
An oracle is said to be an \emph{$\Omega$-oracle} for $H$
if, given any query of the form
$(\omega,h_1,\dots,h_{\ar\omega})$ with $\omega\in\Omega$
and $h_1,\dots,h_{\ar\omega}\in H$, this oracle returns
$\omega(h_1,\dots,h_{\ar\omega})$. (On other queries, the
behavior of the oracle may be arbitrary.)
\end{definition}

If $\Omega$ is a set of group operation symbols, then an
$\Omega$-oracle for a black-box group is called a
\emph{group oracle}. Note that some authors require a group
oracle for a black-box group to perform only the
multiplication and the inversion in this group (see,
e.g.,~\cite{BS84, Bab91, BB93}). It is obvious that the
identity element of any group can be computed as $g^{-1}g$,
where $g$ is an arbitrary element of this group.

\begin{definition}[black-box $\Omega$-algebra algorithm]
\label{d:bbOaalg}
A (possibly probabilistic) algorithm $A$ is called a
\emph{black-box $\Omega$-algebra algorithm} if, when $A$
performs a computation in an arbitrary black-box
$\Omega$-algebra,
\begin{itemize}
\item $A$ has access to an $\Omega$-oracle for this
black-box $\Omega$-algebra and

\item all queries made by $A$ to this $\Omega$-oracle have
the form specified in Definition~\ref{d:Omegaoracle}.
\end{itemize}
\end{definition}

Suppose $A$ is a probabilistic black-box $\Omega$-algebra
algorithm. Consider a computation of $A$ in the black-box
$\Omega$-algebra~$H$. Then Definitions~\ref{d:Omegaoracle}
and~\ref{d:bbOaalg} imply that this computation and its
output depend only on $H$ but not on the $\Omega$-oracle for
$H$ used by~$A$. This is because the answers of this oracle
to the queries made by $A$ are completely determined by~$H$.
Therefore we can denote by $A^H$ the algorithm $A$
performing a computation in $H$ and hence using an
$\Omega$-oracle for~$H$. If the algorithm $A^H$ has access
to an additional oracle, say, $O$, then we denote this
algorithm by~$A^{H,O}$.

Similarly to Subsection~\ref{ss:standmod}, let a black-box
$\Omega$-algebra $H_d$ be assigned to each $d\in D$. When
necessary, we denote by $\pd H_d$ a probability distribution
on the (necessarily nonempty) $\Omega$-algebra $H_d$ for
every $d\in D$. We give analogs of
Definitions~\ref{d:focOawod}, \ref{d:focOa},
\ref{d:wpsfreefocOa}, and~\ref{d:wcwpsfreefocOa} in the
black-box $\Omega$-algebra model. Note that some of these
analogs do not depend on the probability distributions~$\pd
H_d$.

\begin{definition}[family of black-box $\Omega$-algebras
without and with distributions]\leavevmode
\label{d:fobbOawoawd}
\begin{itemize}
\item The family $(H_d\st d\in D)$ is called a \emph{family
of black-box $\Omega$-algebras without distributions} if
there exist a function $\xi\colon D\to\N$ and a polynomial
$\eta$ such that $H_d\subseteq\B^{\xi(d)}$ and
$\xi(d)\le\eta(\acl d)$ for all $d\in D$.

\item The family $((H_d,\pd H_d)\st d\in D)$ is said to be a
\emph{family of black-box $\Omega$-algebras with
distributions} or simply a \emph{family of black-box
$\Omega$-algebras} if $(H_d\st d\in D)$ is a family of
black-box $\Omega$-algebras without distributions.
\end{itemize}
\end{definition}

By default, similarly to families of computational
$\Omega$-algebras, a family of black-box $\Omega$-algebras
is a family with distributions. The main motivation for
introducing the notion of a family of black-box
$\Omega$-algebras without distributions is to abstract from
these distributions whenever this is possible.
Cf.\ Subsection~\ref{ss:standmod}.

\begin{definition}[weakly pseudo-free family of black-box
$\Omega$-algebras]
\label{d:wpsfreefobbOa}
Assume that $((H_d,\pd H_d)\st d\in D)$ is a family of
black-box $\Omega$-algebras in~$\cl V$. Then this family is
called \emph{weakly pseudo-free} in $\cl V$ with respect to
$(\pd D_k\st k\in K)$ and $\sigma$ if for any polynomials
$\pi$ and $\tau$ and any probabilistic polynomial-time
black-box $\Omega$-algebra algorithm~$A$,
\[
\Prob{A^{H_{\rv d}}(1^k,\rv d,\rv g,\rv r)\in\Lambda(H_{\rv
d},\cl V,\sigma,\rv g)}=\negl(k),
\]
where $\rv d\rvd\pd D_k$, $\rv g\rvd\pd H_{\rv d}^{\pi(k)}$,
and $\rv r\rvd\pd H_{\rv d}^{\tau(k)}$.
\end{definition}

In Definition~\ref{d:wpsfreefobbOa}, $\rv r$ is used by $A$
as an additional source of random elements of~$H_{\rv d}$.

\begin{remark}
For a probability distribution $\pd Y$ on $\B^*$, let
$\Smpl\pd Y$ be a probabilistic oracle that returns a random
sample from $\pd Y$ on every query. These samples are chosen
independently of each other regardless of the queries.
Consider the definition obtained from
Definition~\ref{d:wpsfreefobbOa} by removing $\rv r$ (and
$\tau$) and giving the algorithm $A$ access to~$\Smpl\pd
H_{\rv d} $. It is easy to see that this definition is
equivalent to the original one. Namely, assume that
$((H_d,\pd H_d)\st d\in D)$ is a family of black-box
$\Omega$-algebras in $\cl V$, as in
Definition~\ref{d:wpsfreefobbOa}. Then this family is weakly
pseudo-free in $\cl V$ with respect to $(\pd D_k\st k\in K)$
and $\sigma$ if and only if for any polynomial $\pi$ and any
probabilistic polynomial-time black-box $\Omega$-algebra
algorithm~$A$,
\[
\Prob{A^{H_{\rv d},\Smpl\pd H_{\rv d}}(1^k,\rv d,\rv
g)\in\Lambda(H_{\rv d},\cl V,\sigma,\rv g)}=\negl(k),
\]
where $\rv d\rvd\pd D_k$ and $\rv g\rvd\pd H_{\rv
d}^{\pi(k)}$.
\end{remark}

\begin{definition}[worst-case weakly pseudo-free family of
black-box $\Omega$-algebras without distributions]
\label{d:wcwpsfreefobbOa}
Assume that $(H_d\st d\in D)$ is a family of nonempty
black-box $\Omega$-algebras in $\cl V$ without
distributions. Then this family is said to be
\emph{worst-case weakly pseudo-free} in $\cl V$ with respect
to $(D_k\st k\in K)$ and $\sigma$ if for any polynomial
$\pi$ and any probabilistic polynomial-time black-box
$\Omega$-algebra algorithm~$A$,
\[
\min_{d\in D_k\cs g\in
H_d^{\pi(k)}}\Prob{A^{H_d}(1^k,d,g)\in\Lambda(H_d,\cl
V,\sigma,g)}=\negl(k).
\]
\end{definition}

\subsection{Quantum Computation Model}

We assume that the reader is familiar with the basics of
quantum computation. For a detailed introduction to this
model of computation, see~\cite{NC10}, \cite[Part~2]{KSV02},
or~\cite[Section~2 and Appendix~C]{Lom04}.

The purpose of this subsection is to give analogs of
Definitions~\ref{d:wpsfreefocOa}, \ref{d:wcwpsfreefocOa},
\ref{d:Omegaoracle}, \ref{d:bbOaalg}, \ref{d:wpsfreefobbOa},
and~\ref{d:wcwpsfreefobbOa} in the quantum computation
model. For Definitions~\ref{d:wpsfreefocOa} and
\ref{d:wcwpsfreefocOa}, this is straightforward. Namely, it
suffices to require the algorithm $A$ to be quantum.

\begin{definition}[post-quantum weakly pseudo-free family of
computational $\Omega$-algebras] Let $((H_d,\pd H_d)\st d\in
D)$ be a family of computational $\Omega$-algebras in~$\cl
V$. Then this family is called \emph{post-quantum weakly
pseudo-free} in $\cl V$ with respect to $(\pd D_k\st k\in
K)$ and $\sigma$ if for any polynomial $\pi$ and any
polynomial-time quantum algorithm~$A$,
\[
\Prob{A(1^k,\rv d,\rv g)\in\Lambda(H_{\rv d},\cl
V,\sigma,\rv g)}=\negl(k),
\]
where $\rv d\rvd\pd D_k$ and $\rv g\rvd\pd H_{\rv
d}^{\pi(k)}$.
\end{definition}

\begin{definition}[post-quantum worst-case weakly
pseudo-free family of computational $\Omega$-algebras
without distributions] Suppose $(H_d\st d\in D)$ is a family
of nonempty computational $\Omega$-algebras in $\cl V$
without distributions. Then this family is said to be
\emph{post-quantum worst-case weakly pseudo-free} in $\cl V$
with respect to $(D_k\st k\in K)$ and $\sigma$ if for any
polynomial $\pi$ and any polynomial-time quantum
algorithm~$A$,
\[
\min_{d\in D_k\cs g\in
H_d^{\pi(k)}}\Prob{A(1^k,d,g)\in\Lambda(H_d,\cl
V,\sigma,g)}=\negl(k).
\]
\end{definition}

Let $H$ be a black-box $\Omega$-algebra and let $n$ be a
nonnegative integer such that $H\subseteq\B^n$. (If
$H\ne\emptyset$, then $n$ is unique; otherwise $n$ can be
chosen arbitrarily.) We denote by $Q_n$ the state space of
$n$ qubits. Suppose $m\in\N\setminus\{0\}$. Consider a
system of $m$ quantum registers, each consisting of $n$
qubits. Of course, the state space of this system is the
$m$th tensor power of $Q_n$, denoted by~$Q_n^{\otimes m}$.
If for every $i\in\{1,\dots,m\}$ the $i$th quantum register
is in the state $\ket{y_i}\in Q_n$, then we write the state
of the total system as $\ket{y_1}\dots\ket{y_m}$ instead of
$\ket{y_1}\otimes\dots\otimes\ket{y_m}$. (We use the Dirac
ket notation $\ket\cdot$ for quantum state vectors.) For a
unitary operator $W$ on $Q_n^{\otimes r}$ (where
$r\in\{1,\dots,m\}$) and a tuple $(i_1,\dots,i_r)$ of
distinct integers in $\{1,\dots,m\}$, we denote by
$W[i_1,\dots,i_r]$ the unitary operator on $Q_n^{\otimes m}$
acting as $W$ on the system of quantum registers with
numbers $i_1,\dots,i_r$ (taken in this order) and leaving
all other registers unchanged.

\begin{definition}[quantum $\Omega$-oracle]
\label{d:quOmegaoracle}
A family $(U_\omega\st\omega\in\Omega)$, where $U_\omega$ is
a unitary operator on $Q_n^{\otimes((\ar\omega)+1)}$ for
every $\omega\in\Omega$, is called a \emph{quantum
$\Omega$-oracle} for $H$ if
\[
U_\omega(\ket{h_1}\dots\ket{h_{\ar\omega}}\ket v)
=\ket{h_1}\dots\ket{h_{\ar\omega}}\ket{v\oplus
\omega(h_1,\dots,h_{\ar\omega})}
\]
for all $\omega\in\Omega$, $h_1,\dots,h_{\ar\omega}\in H$,
and $v\in\B^n$.
\end{definition}

Similarly to Subsection~\ref{ss:bbOamod}, if $\Omega$ is a
set of group operation symbols, then a quantum
$\Omega$-oracle for a black-box group is called a
\emph{quantum group oracle}.

\begin{remark}
\label{r:qugroracle}
In this remark, we assume that $\Omega$ is a set of group
operation symbols and $H$ is a (black-box) group. In some
works (e.g., in~\cite[Section~2]{Wat01}
and~\cite[Section~2]{IMS03}), a quantum group oracle for $H$
is given by a pair $(M,M')$ of unitary operators on
$Q_n^{\otimes2}$ such that
\begin{equation}
\label{e:condMMpr}
M(\ket g\ket h)=\ket g\ket{gh}\quad\text{and}\quad M'(\ket
g\ket h)=\ket g\ket{g^{-1}h}\quad\text{for all }g,h\in H.
\end{equation}
However, such a pair can be efficiently implemented using a
quantum group oracle for $H$ in the sense of
Definition~\ref{d:quOmegaoracle}, and vice versa. Details
follow.

Let $\CNOT_n$ be the unitary operator on $Q_n^{\otimes2}$
such that $\CNOT_n(\ket v\ket w)=\ket v\ket{v\oplus w}$ for
all $v,w\in\B^n$. Of course, $\CNOT_n$ can be efficiently
implemented by a quantum circuit consisting of $n$
controlled-NOT gates (these gates implement~$\CNOT_1$).

Denote by $\mu$, $\iota$, and $1$ the symbols in $\Omega$
for the multiplication, the inversion, and the identity
element in a group, respectively.

\begin{roenum}
\item Suppose $(U_\mu,U_\iota,U_1)$ is a quantum group
oracle for~$H$. Then
\[
\nreg{\ket g}\ket h\nreg{\ket{gh}}=U_\mu(\ket g\ket
h\ket{0^n})\quad\text{and}\quad\nreg{\ket g}\ket
h\ket{g^{-1}}\nreg{\ket{g^{-1}h}}
=U_\mu[3,2,4]U_\iota[1,3](\ket g\ket h\ket{0^n}\ket{0^n})
\]
for all $g,h\in H$. This yields an efficient implementation
of a pair $(M,M')$ of unitary operators on $Q_n^{\otimes2}$
satisfying condition~\eqref{e:condMMpr}. The contents of the
registers that form the outputs of $M$ and $M'$ are
underlined.

\item Let $(M,M')$ be a pair of unitary operators on
$Q_n^{\otimes2}$ satisfying condition~\eqref{e:condMMpr}.
Then
\begin{align*}
\nreg{\ket g\ket h\ket{v\oplus
gh}}\ket{gh}&=\CNOT_n[4,3]M[1,4]\CNOT_n[2,4](\ket g\ket
h\ket{v}\ket{0^n}),\\\nreg{\ket h\ket{v\oplus
h^{-1}}}\ket{h^{-1}}&=\CNOT_n[3,2](M'[1,3])^2\CNOT_n[1,3]
(\ket h\ket{v}\ket{0^n}),\text{ and}\\\ket
h\nreg{\ket{v\oplus1}}\ket
1&=\CNOT_n[3,2]M'[1,3]\CNOT_n[1,3](\ket h\ket{v}\ket{0^n})
\end{align*}
for all $g,h\in H$ and $v\in\B^n$. This yields an efficient
implementation of a quantum group oracle
$(U_\mu,U_\iota,U_1)$ for~$H$. The contents of the registers
that form the outputs of $U_\mu$, $U_\iota$, and $U_1$ are
underlined.
\end{roenum}
\end{remark}

For each $\omega\in\Omega$, we denote by $E_{H,\omega}$ the
subspace of $Q_n^{\otimes((\ar\omega)+1)}$ spanned by
\[
\{\ket{h_1}\dots\ket{h_{\ar\omega}}\ket v\st
h_1,\dots,h_{\ar\omega}\in H\cs v\in\B^n\}.
\]

\begin{definition}[black-box $\Omega$-algebra quantum
algorithm]
\label{d:bbOaqalg}
A quantum algorithm $A$ is said to be a \emph{black-box
$\Omega$-algebra quantum algorithm} if, when $A$ performs a
computation in an arbitrary black-box $\Omega$-algebra~$G$,
\begin{itemize}
\item $A$ has access to a quantum $\Omega$-oracle (say,
$(U_\omega\st\omega\in\Omega)$) for $G$ and

\item for every $\omega\in\Omega$, the operator $U_\omega$
is applied only to state vectors in~$E_{G,\omega}$.
\end{itemize}
\end{definition}

If $\Omega$ is a set of group operation symbols, then a
black-box $\Omega$-algebra quantum algorithm is called a
\emph{black-box group quantum algorithm} when we are
interested in its computation only in black-box groups.

Suppose $A$ is a black-box $\Omega$-algebra quantum
algorithm. Consider a computation of $A$ in the black-box
$\Omega$-algebra~$H$. Then Definitions~\ref{d:quOmegaoracle}
and~\ref{d:bbOaqalg} imply that this computation and its
output depend only on $H$ but not on the quantum
$\Omega$-oracle for $H$ (say,
$(U_\omega\st\omega\in\Omega)$) used by~$A$. This is because
for every $\omega\in\Omega$, the action of the operator
$U_\omega$ on $E_{H,\omega}$ is completely determined
by~$H$. Therefore, similarly to Subsection~\ref{ss:bbOamod},
we can denote by $A^H$ the algorithm $A$ performing a
computation in $H$ and hence using a quantum $\Omega$-oracle
for~$H$.

\begin{definition}[post-quantum weakly pseudo-free family of
black-box $\Omega$-algebras] Let $((H_d,\pd H_d)\st d\in D)$
be a family of black-box $\Omega$-algebras in~$\cl V$. Then
this family is called \emph{post-quantum weakly pseudo-free}
in $\cl V$ with respect to $(\pd D_k\st k\in K)$ and
$\sigma$ if for any polynomials $\pi$ and $\tau$ and any
polynomial-time black-box $\Omega$-algebra quantum
algorithm~$A$,
\[
\Prob{A^{H_{\rv d}}(1^k,\rv d,\rv g,\rv r)\in\Lambda(H_{\rv
d},\cl V,\sigma,\rv g)}=\negl(k),
\]
where $\rv d\rvd\pd D_k$, $\rv g\rvd\pd H_{\rv d}^{\pi(k)}$,
and $\rv r\rvd\pd H_{\rv d}^{\tau(k)}$.
\end{definition}

\begin{definition}[post-quantum worst-case weakly
pseudo-free family of black-box $\Omega$-algebras without
distributions] Suppose $(H_d\st d\in D)$ is a family of
nonempty black-box $\Omega$-algebras in $\cl V$ without
distributions. Then this family is said to be
\emph{post-quantum worst-case weakly pseudo-free} in $\cl V$
with respect to $(D_k\st k\in K)$ and $\sigma$ if for any
polynomial $\pi$ and any polynomial-time black-box
$\Omega$-algebra quantum algorithm~$A$,
\[
\min_{d\in D_k\cs g\in
H_d^{\pi(k)}}\Prob{A^{H_d}(1^k,d,g)\in\Lambda(H_d,\cl
V,\sigma,g)}=\negl(k).
\]
\end{definition}

\subsection{Relations between the Types of Weak
Pseudo-Freeness}
\label{ss:relations}

In this subsection, weak pseudo-freeness of any type means
weak pseudo-freeness of this type in $\cl V$ with respect to
$(\pd D_k\st k\in K)$ (or $(D_k\st k\in K)$ in the
worst-case setting) and~$\sigma$.

\begin{remark}
\label{r:wpsfreeiswcwpsfree}
In this remark, we assume that $\supp\pd D_k\subseteq D_k$
for all $k\in K$. It is easy to see that if $((H_d,\pd
H_d)\st d\in D)$ is a weakly (resp., post-quantum weakly)
pseudo-free family of computational $\Omega$-algebras, then
$(H_d\st d\in D)$ is a worst-case weakly (resp.,
post-quantum worst-case weakly) pseudo-free family of
computational $\Omega$-algebras without distributions.
Similarly, if $((H_d,\pd H_d)\st d\in D)$ is a weakly
(resp., post-quantum weakly) pseudo-free family of black-box
$\Omega$-algebras, then $(H_d\st d\in D)$ is a worst-case
weakly (resp., post-quantum worst-case weakly) pseudo-free
family of black-box $\Omega$-algebras without distributions.
This is because the minimum of a real-valued random variable
does not exceed the expectation of this random variable,
provided that these minimum and expectation exist.
\end{remark}

\begin{remark}
\label{r:pqwpsfreeiswpsfree}
It is well known that every probabilistic polynomial-time
algorithm can be simulated by a polynomial-time quantum
algorithm (see, e.g.,~\cite[Subsection~1.4.1]{NC10}
or~\cite[Section~7]{KSV02}). Also, every probabilistic
polynomial-time black-box $\Omega$-algebra algorithm can be
simulated by a polynomial-time black-box $\Omega$-algebra
quantum algorithm. Therefore any post-quantum weakly (resp.,
post-quantum worst-case weakly) pseudo-free family of
computational or black-box $\Omega$-algebras with (resp.,
without) distributions is also weakly (resp., worst-case
weakly) pseudo-free.
\end{remark}

\begin{remark}
\label{r:wpffocOaiffwpffobbOa}
Let $\fm H=((H_d,\pd H_d)\st d\in D)$ (resp., $\fm H=(H_d\st
d\in D)$) be a family of computational $\Omega$-algebras in
$\cl V$ with (resp., without) distributions. Assume that
there exist a function $\xi\colon D\to\N$ and a polynomial
$\eta$ such that $H_d\subseteq\B^{\xi(d)}$ and
$\xi(d)\le\eta(\acl d)$ for all $d\in D$. Then by
Definition~\ref{d:fobbOawoawd}, $\fm H$ is a family of
black-box $\Omega$-algebras with (resp., without)
distributions. Furthermore, $\fm H$ is weakly pseudo-free or
post-quantum weakly pseudo-free (resp., worst-case weakly
pseudo-free or post-quantum worst-case weakly pseudo-free)
as a family of computational $\Omega$-algebras with (resp.,
without) distributions if and only if $\fm H$ satisfies the
same weak pseudo-freeness condition as a family of black-box
$\Omega$-algebras with (resp., without) distributions. This
can be proved straightforwardly.
\end{remark}

\begin{proposition}
\label{p:exwpffocOaimplexwpffobbOa}
Assume that there exists a weakly pseudo-free or a
post-quantum weakly pseudo-free (resp., a worst-case weakly
pseudo-free or a post-quantum worst-case weakly pseudo-free)
family of computational $\Omega$-algebras with (resp.,
without) distributions. Then there exists a family of
black-box $\Omega$-algebras with (resp., without)
distributions that satisfies the same weak pseudo-freeness
condition as a family of black-box $\Omega$-algebras with
(resp., without) distributions.
\end{proposition}
\begin{proof}
For any $n\in\N$, let $\alpha_n$ be the one-to-one function
from $\B^{\le n}$ onto $\B^{n+1}\setminus\{0^{n+1}\}$
defined by $\alpha_n(u)=u10^{n-\acl u}$ for all $u\in\B^{\le
n}$. (Here, of course, $u10^{n-\acl u}$ denotes the
concatenation of $u$, $1$, and $0^{n-\acl u}$.) Then the
functions $(1^n,u)\mapsto\alpha_n(u)$ and
$(1^n,t)\mapsto\alpha_n^{-1}(t)$, where $n\in\N$,
$u\in\B^{\le n}$, and $t\in\B^{n+1}\setminus\{0^{n+1}\}$,
are polynomial-time computable. (Note that $1^n$ in
$(1^n,t)$ is redundant; we write it only for the reader's
convenience.)

Suppose $\fm G=((G_d,\pd G_d)\st d\in D)$ is a family of
computational $\Omega$-algebras in~$\cl V$. Choose a
polynomial $\eta$ such that $G_d\subseteq\B^{\le\eta(\acl
d)}$ for all $d\in D$. For each such $d$, let
$H_d=\alpha_{\eta(\acl d)}(G_d)\subseteq\B^{\eta(\acl
d)+1}\setminus\{0^{\eta(\acl d)+1}\}$ and $\pd
H_d=\alpha_{\eta(\acl d)}(\pd G_d)$. Consider $H_d$ as the
unique $\Omega$-algebra such that the restriction of
$\alpha_{\eta(\acl d)}$ to $G_d$ is an isomorphism of $G_d$
onto~$H_d$. Then it is easy to see that $\fm H=((H_d,\pd
H_d)\st d\in D)$ is a family of computational
$\Omega$-algebras in~$\cl V$. Moreover, $\fm H$ is also a
family of black-box $\Omega$-algebras (see
Remark~\ref{r:wpffocOaiffwpffobbOa}). This is because there
exists a polynomial $\eta'$ such that $\eta(n)+1\le\eta'(n)$
for all $n\in\N$.

Assume that $\fm G$ is a weakly pseudo-free family of
computational $\Omega$-algebras. It is easy to show that for
any isomorphic $\Omega$-algebras $G,H\in\cl V$, any $g\in
G^m$, where $m\in\N\setminus\{0\}$, and any isomorphism
$\alpha\colon G\to H$, we have $\Lambda(H,\cl
V,\sigma,\alpha(g))=\Lambda(G,\cl V,\sigma,g)$. This implies
that $\fm H$ is a weakly pseudo-free family of computational
$\Omega$-algebras. Indeed, suppose $\pi$ is a polynomial and
$A$ is a probabilistic polynomial-time algorithm trying to
break the weak pseudo-freeness of $\fm H$ for $\pi$ (i.e.,
to violate the condition of Definition~\ref{d:wpsfreefocOa}
for $\fm H$ and~$\pi$). Let $B$ be a probabilistic
polynomial-time algorithm (trying to break the weak
pseudo-freeness of $\fm G$ for $\pi$) that on input
$(1^k,d,g)$ for every $k\in K$, $d\in\supp\pd D_k$, and
$g\in(\supp\pd G_d)^{\pi(k)}$ runs $A$ on input
$(1^k,d,\alpha_{\eta(\acl d)}(g))$ and returns the output
(if it exists). Then
\begin{align*}
\Prob{A(1^k,\rv d,\rv h)\in\Lambda(H_{\rv d},\cl
V,\sigma,\rv h)}&=\Prob{A(1^k,\rv d,\alpha_{\eta(\acl{\rv
d})}(\rv g))\in\Lambda(H_{\rv d},\cl
V,\sigma,\alpha_{\eta(\acl{\rv d})}(\rv g))}\\
&=\Prob{B(1^k,\rv d,\rv g)\in\Lambda(G_{\rv d},\cl
V,\sigma,\rv g)}=\negl(k),
\end{align*}
where $\rv d\rvd\pd D_k$, $\rv h\rvd\pd H_{\rv d}^{\pi(k)}$,
and $\rv g\rvd\pd G_{\rv d}^{\pi(k)}$. Here we use the fact
that the random variables $(\rv d,\rv h)$ and $(\rv
d,\alpha_{\eta(\acl{\rv d})}(\rv g))$ are identically
distributed.

By Remark~\ref{r:wpffocOaiffwpffobbOa}, $\fm H$ is also a
weakly pseudo-free family of black-box $\Omega$-algebras.
Thus, if there exists a weakly pseudo-free family of
computational $\Omega$-algebras, then there exists a weakly
pseudo-free family of black-box $\Omega$-algebras. For other
types of weak pseudo-freeness mentioned in the proposition,
the proofs are the same, \latin{mutatis mutandis}.
\end{proof}

We illustrate the statements of
Remarks~\ref{r:wpsfreeiswcwpsfree}
and~\ref{r:pqwpsfreeiswpsfree} and of
Proposition~\ref{p:exwpffocOaimplexwpffobbOa} by the diagram
in Figure~\ref{f:relations}.

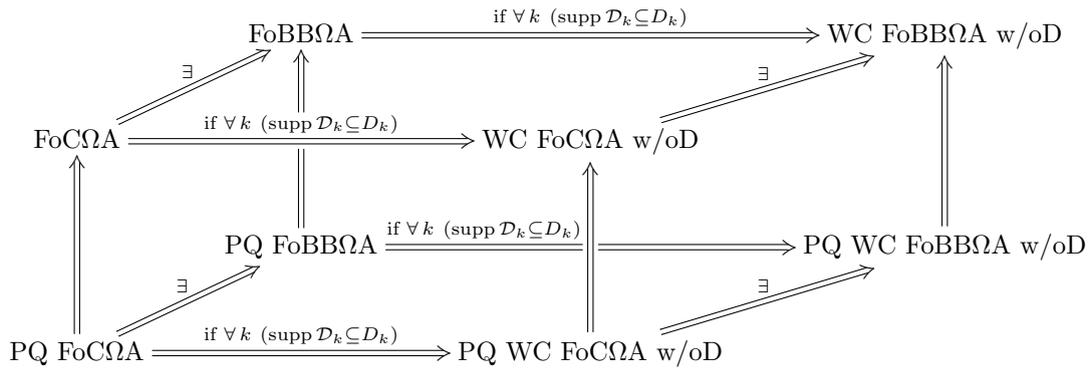
\begin{figure}[htb]
\[
\xymatrix{&\text{FoBB$\Omega$A}\impl[rr]^-\Dk&&
\text{WC FoBB$\Omega$A w/oD}\\\text{FoC$\Omega$A}
\impl[ur]^-\exists\impl[rr]^-\Dk&&
\text{WC FoC$\Omega$A w/oD}\impl[ur]^-\exists&\\
&\text{PQ FoBB$\Omega$A}\impl[uu]|(.55)*\txt{\\\hole}
\impl'[r]^-\Dk[rr]&&\text{PQ WC FoBB$\Omega$A w/oD}
\impl[uu]\\\text{PQ FoC$\Omega$A}\impl[uu]\impl[ur]^-\exists
\impl[rr]^-\Dk&&\text{PQ WC FoC$\Omega$A w/oD}
\impl[uu]\impl[ur]^-\exists&}
\]
\caption{Relations between the types of weak pseudo-freeness
defined in Section~\ref{s:wpsfreefams}. The abbreviations
PQ, WC, FoC$\Omega$A, FoBB$\Omega$A, and w/oD stand for
Post-Quantum, Worst-Case, (weakly pseudo-free) Family of
Computational $\Omega$-Algebras, (weakly pseudo-free) Family
of Black-Box $\Omega$-Algebras, and without Distributions,
respectively. For brevity, we do not write an abbreviation
for Weakly Pseudo-Free in the diagram. Weak pseudo-freeness
of any type means weak pseudo-freeness of this type in $\cl
V$ with respect to $(\pd D_k\st k\in K)$ (or $(D_k\st k\in
K)$ in the worst-case setting) and~$\sigma$. A horizontal
double-line arrow from $Y$ to $Z$ labeled ``$\Dk$'' means
that if $((H_d,\pd H_d)\st d\in D)$ is a family of type $Y$
and $\supp\pd D_k\subseteq D_k$ for all $k\in K$, then
$(H_d\st d\in D)$ is a family of type $Z$ (see
Remark~\ref{r:wpsfreeiswcwpsfree}). Furthermore, a vertical
unlabeled double-line arrow from $Y$ to $Z$ means that any
family of type $Y$ is also a family of type $Z$ (see
Remark~\ref{r:pqwpsfreeiswpsfree}). Finally, an oblique
double-line arrow from $Y$ to $Z$ labeled ``$\exists$''
means that if there exists a family of type $Y$, then there
exists a family of type $Z$ (see
Proposition~\ref{p:exwpffocOaimplexwpffobbOa}).}
\label{f:relations}
\end{figure}

\subsection{Weak Pseudo-Freeness in the Variety Generated by
the \headPsi-Reducts of All \headOmega-Algebras in~\headclV}
\label{ss:wpffofreducts}

Let $\Psi$ be a subset of $\Omega$ and let $H$ be an
$\Omega$-algebra. Then the $\Psi$-algebra obtained from $H$
by omitting the fundamental operations associated with the
symbols in $\Omega\setminus\Psi$ is called the
\emph{$\Psi$-reduct} of $H$ (or the \emph{reduct} of $H$ to
$\Psi$). We denote the $\Psi$-reduct of $H$ by $H|_\Psi$.
The $\Omega$-algebra $H$ is said to be an \emph{expansion}
of $H|_\Psi$ to~$\Omega$. Furthermore, $\cl V|_\Psi$ denotes
the variety of $\Psi$-algebras generated by the
$\Psi$-reducts of all $\Omega$-algebras in~$\cl V$. In other
words, $\cl V|_\Psi$ is the variety of $\Psi$-algebras
defined by the set of all identities over $\Psi$ that hold
in $\cl V$ (actually, in $G|_\Psi$ for every $G\in\cl V$).
Clearly, $\cl V|_\Psi$ is nontrivial if and only if $\cl V$
is nontrivial.

The $\Omega$-algebra $H$ is called an \emph{expanded group}
if there exists a set $\Gamma\subseteq\Omega$ of group
operation symbols such that $H|_\Gamma$ is a group. When it
comes to classes of expanded groups, we assume that this set
$\Gamma$ is the same for all expanded groups in the class.
Thus, $\cl V$ is said to be a \emph{variety of expanded
groups} if $\Omega$ contains a set $\Gamma$ of group
operation symbols such that $H|_\Gamma$ is a group for all
$H\in\cl V$ (or, equivalently, $\cl V|_\Gamma$ is a variety
of groups).

Lemma~\ref{l:freegensys} implies that the subalgebra of
$F_\infty(\cl V)|_\Psi$ generated by $a_1,a_2,\dotsc$ is a
$\cl V|_\Psi$-free $\Psi$-algebra freely generated by this
system. So we choose this $\Psi$-algebra as $F_\infty(\cl
V|_\Psi)$. Similarly, we assume that for any $m\in\N$,
$F_m(\cl V|_\Psi)$ is the subalgebra of $F_m(\cl V)|_\Psi$
generated by $a_1,\dots,a_m$. In particular, $F_\infty(\cl
V|_\Psi)\subseteq F_\infty(\cl V)$ and $F_m(\cl
V|_\Psi)\subseteq F_m(\cl V)$ for all $m\in\N$.

In the next proposition, we assume that an $\Omega$-algebra
$H_d\subseteq\B^*$ is assigned to each $d\in D$. When
necessary, we denote by $\pd H_d$ a probability distribution
on the (necessarily nonempty) $\Omega$-algebra $H_d$ for
every $d\in D$. Cf.\ Subsections~\ref{ss:standmod}
and~\ref{ss:bbOamod}.

\begin{proposition}
\label{p:wpffofreducts}
Suppose $S$ is a subset of $\dom\sigma$ such that
$\sigma(S)=F_\infty(\cl V|_\Psi)$. Let $\fm H=((H_d,\pd
H_d)\st d\in D)$ (resp., $\fm H=(H_d\st d\in D)$) and $\fm
H'=((H_d|_\Psi,\pd H_d)\st d\in D)$ (resp., $\fm
H'=(H_d|_\Psi\st d\in D)$). Then the following statements
hold:
\begin{roenum}
\item\label{i:wpffofreducts:comp} If $\fm H$ is a family of
computational $\Omega$-algebras in $\cl V$ with (resp.,
without) distributions, then $\fm H'$ is a family of
computational $\Psi$-algebras in $\cl V|_\Psi$ with (resp.,
without) distributions.

\item\label{i:wpffofreducts:bb} If $\fm H$ is a family of
black-box $\Omega$-algebras in $\cl V$ with (resp., without)
distributions, then $\fm H'$ is a family of black-box
$\Psi$-algebras in $\cl V|_\Psi$ with (resp., without)
distributions.

\item\label{i:wpffofreducts:wpsfree} If $\fm H$ is a weakly
pseudo-free or post-quantum weakly pseudo-free (resp.,
worst-case weakly pseudo-free or post-quantum worst-case
weakly pseudo-free) family of computational or black-box
$\Omega$-algebras (with (resp., without) distributions) in
$\cl V$ with respect to $(\pd D_k\st k\in K)$ (resp.,
$(D_k\st k\in K)$) and $\sigma$, then $\fm H'$ satisfies the
same weak pseudo-freeness condition in $\cl V|_\Psi$ with
respect to $(\pd D_k\st k\in K)$ (resp., $(D_k\st k\in K)$)
and~$\sigma|_S$.
\end{roenum}
\end{proposition}
\begin{proof}
Statements~\ref{i:wpffofreducts:comp}
and~\ref{i:wpffofreducts:bb} can be proved
straightforwardly. Suppose $H\in\cl V$ and $g\in H^m$, where
$m\in\N\setminus\{0\}$. Then it is easy to show that
$\Lambda(H|_\Psi,\cl
V|_\Psi,\sigma|_S,g)\subseteq\Lambda(H,\cl V,\sigma,g)$.
(Note that $\Lambda(H|_\Psi,\cl
V|_\Psi,\sigma|_S,g)=\Lambda(H,\cl V,\sigma,g)\cap S^2$, but
we do not need this fact.) Furthermore, a black-box
$\Psi$-algebra algorithm can be considered as a black-box
$\Omega$-algebra algorithm. The same holds for quantum
algorithms. These observations imply
statement~\ref{i:wpffofreducts:wpsfree}.
\end{proof}

\begin{remark}
\label{r:appltoSLP}
In particular, Proposition~\ref{p:wpffofreducts} can be
applied to the case where $\sigma=\SLP$ and $S$ is the set
of all straight-line programs over~$\Psi$. For this set $S$,
we have $\mathord{\SLP}|_S=\SLP_{\cl V|_\Psi}$.
\end{remark}

\section{Some Polynomial-Time Black-Box Group Quantum
Algorithms}
\label{s:ptbbgrqalgs}

In this section, we assume that $\Omega$ is a set of group
operation symbols and $\cl V$ is a variety of groups. We
prove that if $\cl V$ is nontrivial, then there exists a
polynomial-time black-box group quantum algorithm $B$ such
that for any black-box group $G\in\cl V$ and any $g\in G^m$
with $m>\log_2\acl G$, we have $\Prob{B^G(g)\in\Lambda(G,\cl
V,\SLP,g)}\ge\epsilon$, where $\epsilon$ is a positive
constant. See Lemmas~\ref{l:ProutBinLambdainfexp}
and~\ref{l:ProutBinLambdantrnall} below.

\subsection{The Case Where \headclV\ Has Infinite Exponent}

Throughout this subsection, we assume that the variety $\cl
V$ is of infinite exponent and that, given
$s\in\N\setminus\{0\}$, one can compute $\rep{a_1^s}\sigma$
in polynomial time. Of course, $\SLP$ satisfies the latter
assumption.

In the next lemma, we denote by $A$ a polynomial-time
black-box group quantum algorithm such that for any
black-box group $G\in\cl V$ and any $g\in G$,
\begin{equation}
\label{e:condAinfexp}
\Prob{A^G(g)=s\in\N\setminus\{0\}\text{
s.t.\ }g^s=1}\ge\epsilon,
\end{equation}
where $\epsilon$ is a positive constant. Such an algorithm
exists. For example, Shor's order-finding algorithm
(see~\cite[Section~5]{Sho97}, \cite[Subsection~5.3.1]{NC10},
or~\cite[Subsections~13.4--13.6]{KSV02}) can be easily
converted to a polynomial-time black-box group quantum
algorithm that satisfies the required condition.

\begin{lemma}
\label{l:ProutBinLambdainfexp}
There exists a polynomial-time black-box group quantum
algorithm $B$ such that for any black-box group $G\in\cl V$
and any $g\in G$,
\[
\Prob{B^G(g)\in\Lambda(G,\cl V,\sigma,g)}\ge\epsilon,
\]
where $\epsilon$ is the same positive constant as
in~\eqref{e:condAinfexp}.
\end{lemma}
\begin{proof}
Let $B$ be a polynomial-time black-box group quantum
algorithm such that for any black-box group $G\in\cl V$ and
any $g\in G$, $B$ on input $g$ with access to a quantum
group oracle for $G$ proceeds as follows:
\begin{arenum}
\item Run $A^G$ on input~$g$.

\item If the output is a positive integer $s$ satisfying
$g^s=1$, then return $(\rep{a_1^s}\sigma,\rep1\sigma)$.
Otherwise the algorithm $B$ fails.
\end{arenum}
Suppose $G$ and $g$ are as in the statement of the lemma.
Then it is easy to see that $B^G(g)\in\Lambda(G,\cl
V,\sigma,g)$ if and only if $A^G(g)=s$, where
$s\in\N\setminus\{0\}$ and $g^s=1$. (Note that $a_1^s\ne1$
for all $s\in\N\setminus\{0\}$ because $\cl V$ has infinite
exponent.) Hence,
\[
\Prob{B^G(g)\in\Lambda(G,\cl
V,\sigma,g)}=\Prob{A^G(g)=s\in\N\setminus\{0\}\text{
s.t.\ }g^s=1}\ge\epsilon.\qed
\]
\renewcommand\qed{}
\end{proof}

\subsection{The Case Where \headclV\ Is Nontrivial and
Is Not the Variety of All Groups}

Throughout this subsection, we assume that $\cl V$ is
nontrivial and is not the variety of all groups. Moreover,
all results of this subsection depend on the CFSG unless
$\cl V$ is solvable.

Let $H$ be a finite group. As in~\cite{BB93}
and~\cite{IMS03}, we denote by $\nu(H)$ the smallest
$n\in\N\setminus\{0\}$ such that all nonabelian composition
factors of $H$ can be embedded in the symmetric group of
degree~$n$. If all composition factors of $H$ are abelian
(i.e., $H$ is solvable), then $\nu(H)=1$.

The \emph{constructive membership problem} for subgroups of
$H$ is defined as follows: Given $g_1,\dots,g_m,h\in H$
(where $m\in\N$), either find a straight-line program
computing $h$ from $g_1,\dots,g_m$ (if
$h\in\alg{g_1,\dots,g_m}$) or report that no such
straight-line program exists (if
$h\notin\alg{g_1,\dots,g_m}$). Lemma~\ref{l:shortSLP}
implies that if $h\in\alg{g_1,\dots,g_m}$, then there exists
such a straight-line program of length at most
$(1+\log_2\acl H)^2$.

\begin{remark}
\label{r:ptqalgforCMP}
Ivanyos, Magniez, and Santha~\cite[Theorem~5]{IMS03} proved
the existence of a black-box group quantum algorithm that
solves the constructive membership problem for subgroups of
any given black-box group $G$ in time polynomial in the
input length${}+\nu(G)$ with success probability at least
$\epsilon$, where $\epsilon$ is a constant satisfying
$1/2<\epsilon<1$. The proof of this uses a deep algorithmic
result of Beals and Babai~\cite[Theorem~1.2]{BB93}, which in
turn depends on the CFSG\@. Note that in~\cite{IMS03},
quantum group oracles for black-box groups and straight-line
programs for groups are defined slightly differently than in
the present paper. However, this does not matter for us.
Indeed, by Remark~\ref{r:qugroracle}, a quantum group oracle
in the sense of~\cite{IMS03} for an arbitrary black-box
group can be efficiently implemented using a quantum group
oracle in the sense of the present paper for that black-box
group, and vice versa. Also, if $g_1,\dots,g_m$ are elements
of a group, where $m\in\N\setminus\{0\}$, and
$h\in\alg{g_1,\dots,g_m}$, then a straight-line program
computing $h$ from $g_1,\dots,g_m$ in the sense
of~\cite{IMS03} can be efficiently converted to a
straight-line program computing $h$ from $g_1,\dots,g_m$ in
the sense of the present paper, and vice versa.

Furthermore, by a result of Jones~\cite{Jon74} together with
the CFSG, there are only finitely many (up to isomorphism)
nonabelian finite simple groups in every variety of groups
different from the variety of all groups. (This gives a
negative answer to Problem~23 in~\cite{Neu67}.) Hence for
any finite group $H\in\cl V$, $\nu(H)$ is upper bounded by a
constant because $\cl V$ is not the variety of all groups.
Thus, the above-mentioned black-box group quantum algorithm
for the constructive membership problem runs in polynomial
time whenever the given black-box group is in~$\cl V$.

Note that if $\cl V$ is solvable, then $\nu(H)=1$ for every
finite group $H\in\cl V$ and the above-mentioned algorithm
of Ivanyos, Magniez, and Santha does not deal with
nonabelian finite simple groups during a computation in a
black-box group in~$\cl V$. Therefore in this case, we do
not need the CFSG for our purposes.
\end{remark}

In the next lemma, we denote by $A$ a polynomial-time
black-box group quantum algorithm such that for any
black-box group $G\in\cl V$, any $g_1,\dots,g_m\in G$ (where
$m\in\N$), and any $h\in\alg{g_1,\dots,g_m}$,
\begin{equation}
\label{e:condAntrnall}
\Prob{A^G(g_1,\dots,g_m,h)\text{ is a straight-line program
computing }h\text{ from }g_1,\dots,g_m}\ge\epsilon,
\end{equation}
where $\epsilon$ is a positive constant. By
Remark~\ref{r:ptqalgforCMP}, such an algorithm exists.

\begin{lemma}
\label{l:ProutBinLambdantrnall}
There exists a polynomial-time black-box group quantum
algorithm $B$ such that for any black-box group $G\in\cl V$
and any $g\in G^m$ with $m>\log_2\acl G$,
\[
\Prob{B^G(g)\in\Lambda(G,\cl V,\SLP,g)}\ge\epsilon,
\]
where $\epsilon$ is the same positive constant as
in~\eqref{e:condAntrnall}.
\end{lemma}
\begin{proof}
Let $B$ be a polynomial-time black-box group quantum
algorithm such that for any black-box group $G\in\cl V$ and
any $g=(g_1,\dots,g_m)$, where $m\in\N$ and
$g_1,\dots,g_m\in G$, $B$ on input $g$ with access to a
quantum group oracle for $G$ proceeds as follows:
\begin{arenum}
\item For each $i\in\{1,\dots,m\}$ (in ascending order), run
$A^G$ on input $(g_1,\dots,g_i)$. If the output is a
straight-line program $u$ that computes $g_i$ from
$g_1,\dots,g_{i-1}$, then return $(u,(i))$ and stop. (The
straight-line program $(i)$ computes the $i$th element of
the input sequence.)

\item If this point is reached, then the algorithm $B$
fails.
\end{arenum}
Suppose $G$ and $g=(g_1,\dots,g_m)$ are as in the statement
of the lemma (in particular, $m>\log_2\acl G$). If
$g_i\notin\alg{g_1,\dots,g_{i-1}}$ for all
$i\in\{1,\dots,m\}$, then an easy induction on $i$ shows
that $\acl{\alg{g_1,\dots,g_i}}\ge2^i$ for every
$i\in\{0,\dots,m\}$. In particular,
$2^m\le\acl{\alg{g_1,\dots,g_m}}\le\acl G$, which
contradicts $m>\log_2\acl G$. Hence
$g_j\in\alg{g_1,\dots,g_{j-1}}$ for some
$j\in\{1,\dots,m\}$. Choose the smallest such~$j$. Assume
that $A^G(g_1,\dots,g_j)=u$, where $u$ is a straight-line
program computing $g_j$ from $g_1,\dots,g_{j-1}$; this holds
with probability at least~$\epsilon$. Let $v=\SLP(u)$, i.e.,
$v$ is the element of $F_{j-1}(\cl V)$ computed from
$a_1,\dots,a_{j-1}$ by~$u$. Then
$v=v(a)=v(a_1,\dots,a_{j-1})\ne a_j$ (because $\cl V$ is
nontrivial) and $v(g)=v(g_1,\dots,g_{j-1})=g_j$. This shows
that $B^G(g)=(u,(j))=(\rep
v\SLP,\rep{a_j}\SLP)\in\Lambda(G,\cl V,\SLP,g)$ with
probability at least~$\epsilon$.
\end{proof}

\section{Main Result}
\label{s:mainres}

\begin{theorem}
\label{t:tharenopqwpffoeg}
Assume that $\cl V$ is a nontrivial variety of expanded
groups. Then there are no families of any of the following
types:
\begin{roenum}
\item\label{i:tharenopqwpffoeg:pqwpffocOa} post-quantum
weakly pseudo-free families of computational
$\Omega$-algebras in $\cl V$ with respect to $(\pd D_k\st
k\in K)$ and~$\SLP$,

\item post-quantum worst-case weakly pseudo-free families of
computational $\Omega$-algebras (without distributions) in
$\cl V$ with respect to $(D_k\st k\in K)$ and~$\SLP$,

\item\label{i:tharenopqwpffoeg:pqwpffobbOa} post-quantum
weakly pseudo-free families of black-box $\Omega$-algebras
in $\cl V$ with respect to $(\pd D_k\st k\in K)$ and~$\SLP$,

\item\label{i:tharenopqwpffoeg:pqwcwpffobbOa} post-quantum
worst-case weakly pseudo-free families of black-box
$\Omega$-algebras (without distributions) in $\cl V$ with
respect to $(D_k\st k\in K)$ and~$\SLP$.
\end{roenum}
\end{theorem}
\begin{proof}
Choose a set $\Gamma\subseteq\Omega$ of group operation
symbols such that $\cl V|_\Gamma$ is a variety of groups.
Remark~\ref{r:wpsfreeiswcwpsfree} and
Proposition~\ref{p:exwpffocOaimplexwpffobbOa} (see also
Figure~\ref{f:relations}) show that it is sufficient to
prove the nonexistence of families of
type~\ref{i:tharenopqwpffoeg:pqwcwpffobbOa} (for families of
types~\ref{i:tharenopqwpffoeg:pqwpffocOa}
and~\ref{i:tharenopqwpffoeg:pqwpffobbOa}, we put
$D_k=\supp\pd D_k$ for all $k\in K$). Furthermore,
Proposition~\ref{p:wpffofreducts} and
Remark~\ref{r:appltoSLP} imply that for this it suffices to
prove the nonexistence of post-quantum worst-case weakly
pseudo-free families of black-box groups (without
distributions) in $\cl V|_\Gamma$ with respect to $(D_k\st
k\in K)$ and~$\SLP_{\cl V|_\Gamma}$. (Of course, the variety
$\cl V|_\Gamma$ is nontrivial.)

Suppose $(G_d\st d\in D)$ is a family of black-box groups
(without distributions) in $\cl V|_\Gamma$, where
$G_d\subseteq\B^{\xi(d)}$ for each $d\in D$ ($\xi\colon
D\to\N$). Let $B$ be a polynomial-time black-box group
quantum algorithm from either
Lemma~\ref{l:ProutBinLambdainfexp} if the exponent of $\cl
V|_\Gamma$ is infinite or
Lemma~\ref{l:ProutBinLambdantrnall} otherwise. Also, suppose
$\epsilon$ is the positive constant from that lemma. If the
exponent of $\cl V|_\Gamma$ is infinite, then let $\pi$ be
the constant polynomial $n\mapsto1$ ($n\in\N$). Otherwise
choose a polynomial $\pi$ such that $\xi(d)<\pi(k)$ for all
$k\in K$ and $d\in D_k$. Such a polynomial exists because
there are polynomials $\eta$ and $\theta$ satisfying
$\xi(d)\le\eta(\acl d)$ for all $d\in D$ and $\acl
d\le\theta(k)$ for all $k\in K$ and $d\in D_k$.

Let $C$ be a polynomial-time black-box group quantum
algorithm such that for any $k\in K$, $d\in D_k$, and $g\in
G_d^{\pi(k)}$, $C$ on input $(1^k,d,g)$ with access to a
quantum group oracle for $G_d$ runs $B^{G_d}$ on input $g$
and returns the output (if it exists). Suppose $k$, $d$, and
$g$ are as in the previous sentence. Note that if the
exponent of $\cl V|_\Gamma$ is finite, then
$\acl{G_d}\le2^{\xi(d)}$ and hence
$\pi(k)>\xi(d)\ge\log_2\acl{G_d}$. By either
Lemma~\ref{l:ProutBinLambdainfexp} (if the exponent of $\cl
V|_\Gamma$ is infinite) or
Lemma~\ref{l:ProutBinLambdantrnall} (otherwise), we have
\[
\Prob{C^{G_d}(1^k,d,g)\in\Lambda(G_d,\cl V|_\Gamma,\SLP_{\cl
V|_\Gamma},g)}\ge\epsilon.
\]
Therefore,
\[
\min_{d\in D_k\cs g\in
G_d^{\pi(k)}}\Prob{C^{G_d}(1^k,d,g)\in\Lambda(G_d,\cl
V|_\Gamma,\SLP_{\cl V|_\Gamma},g)}\ge\epsilon
\]
for all $k\in K$. This shows that $(G_d\st d\in D)$ is not
post-quantum worst-case weakly pseudo-free in $\cl
V|_\Gamma$ with respect to $(D_k\st k\in K)$ and~$\SLP_{\cl
V|_\Gamma}$. Thus, there are no post-quantum worst-case
weakly pseudo-free families of black-box groups (without
distributions) in $\cl V|_\Gamma$ with respect to $(D_k\st
k\in K)$ and~$\SLP_{\cl V|_\Gamma}$.
\end{proof}

Note that if the set $\Gamma$ in the proof of
Theorem~\ref{t:tharenopqwpffoeg} cannot be chosen so that
$\cl V|_\Gamma$ has infinite exponent or is solvable, then
the statement of this theorem depends on the CFSG\@.

\begin{remark}
\label{r:apploftheorem}
In particular, Theorem~\ref{t:tharenopqwpffoeg} can be
applied to nontrivial varieties of groups, rings, modules
and algebras over a finitely generated commutative
associative ring with $1$, near-rings, and, more generally,
groups with finitely many multiple operators.
See~\cite{Hig56} or \cite[Chapter~II, Section~2]{Cohn81} for
the definition of a group with multiple operators (also
known as a multi-operator group).
\end{remark}

\section{Conclusion and Directions for Future Research}
\label{s:conclanddirforfutres}

We have shown that in any nontrivial variety of expanded
groups, there are no post-quantum weakly pseudo-free
families with respect to $(\pd D_k\st k\in K)$ (or $(D_k\st
k\in K)$ in the worst-case setting) and $\SLP$, even in the
worst-case setting and/or the black-box model. In our
opinion, this is an additional motivation for studying
(weak) pseudo-freeness in varieties of $\Omega$-algebras
that in general are not expanded groups. In particular, it
would be interesting to explore (weakly) pseudo-free
families of semigroups, monoids, equasigroups, eloops, and
lattices. The terms ``equasigroup'' and ``eloop'' mean
``equationally definable quasigroup'' and ``equationally
definable loop,'' respectively; see \cite[Chapter~IV,
Section~1]{Cohn81} for definitions of these terms.

Here are some open questions for future research:
\begin{itemize}
\item Does the statement of Theorem~\ref{t:tharenopqwpffoeg}
hold if $\cl V$ is a nontrivial variety of semigroups,
monoids, equasigroups, eloops, or lattices?

\item Does Theorem~\ref{t:tharenopqwpffoeg} remain valid if
the families of computational and black-box
$\Omega$-algebras are polynomially bounded but do not
necessarily have unique representations of elements?

\item Do the types of weak pseudo-freeness defined in
Section~\ref{s:wpsfreefams} (and/or the respective types of
pseudo-freeness) have interesting properties?

\item Does there exist an exponential-size post-quantum
(weakly) pseudo-free family of computational groups in some
nontrivial variety $\cl V$ of groups under a standard
cryptographic assumption? We do not require this family to
be polynomially bounded or to have unique representations of
elements. Also, one may use any natural representation of
elements of the $\cl V$-free group by bit strings.

\item Theorem~4.2 in~\cite{Ano13} states that under the
general integer factoring intractability assumption, a
certain family (say, $\fm G$) of computational groups is
pseudo-free in the variety of all groups with respect to a
certain probability ensemble $(\pd E_k\st k\in K)$ and a
certain function~$\beta$. The family $\fm G$ has exponential
size, but is not polynomially bounded and does not have
unique representations of elements. Is $\fm G$ post-quantum
weakly pseudo-free in the variety of all groups with respect
to $(\pd E_k\st k\in K)$ and~$\beta$? The conjectured answer
is no.
\end{itemize}

\bibliographystyle{alpha}\bibliography{pseudo-free}

\appendix

\section{Table of Notation}
\label{a:tabofnot}

In this appendix, for the convenience of the reader, we
briefly recall some notation introduced in
Sections~\ref{s:prelim} and~\ref{s:wpsfreefams} (in order of
appearance).

{\renewcommand{\arraystretch}{1.3}
\begin{longtable}{l@{\extracolsep{\fill}}p{.8\textwidth}}
$\N$&$=\{0,1,\dotsc\}$\\

$Y^n$&the set of all (ordered) $n$-tuples of elements
from~$Y$\\

$Y^{\le n}$&$=\bigcup_{i=0}^nY^i$\\

$Y^*$&$=\bigcup_{i=0}^\infty Y^i$\\

$\acl u$&the length of bit string~$u$\\

$1^n$&the unary representation of $n\in\N$, i.e., the string
of $n$ ones\\

$0^n$&the string of $n$ zeros\\

$\oplus$&the bitwise XOR operation\\

$(q_i\st i\in I)$&the family of objects $q_i$ ($i\in I$)\\

$\dom\phi$&the domain of function~$\phi$\\

$\id_Y$&the identity function on~$Y$\\

$\rep t\rho$&an arbitrary preimage of $t$ under
function~$\rho$\\

$\R_+$&the set of all nonnegative real numbers\\

$\Omega$&a set of finitary operation symbols (in all
sections except Section~\ref{s:prelim}, $\Omega$ is finite;
in Section~\ref{s:ptbbgrqalgs}, $\Omega$ is a set of group
operation symbols)\\

$\ar\omega$&the arity of $\omega\in\Omega$\\

$\alg S$&the subalgebra generated by~$S$\\

$\Omega_n$&the set of all $n$-ary operation symbols
in~$\Omega$\\

$\Tm Z$&the $\Omega$-term algebra over~$Z$\\

$\cl V$&a variety of $\Omega$-algebras (in
Section~\ref{s:ptbbgrqalgs}, $\cl V$ is a variety of
groups)\\

$F_\infty(\cl V)$&the $\cl V$-free $\Omega$-algebra freely
generated by $a_1,a_2,\dotsc$\\

$F_m(\cl V)$&$=\alg{a_1,\dots,a_m}$\\

$v(a)$&$=v(a_1,\dots,a_m)$ for $v\in F_m(\cl V)$\\

$v(g)$&$=v(g_1,\dots,g_m)$ for $v\in F_m(\cl V)$ and
$g=(g_1,\dots,g_m)\in G^m$, where $G\in\cl V$\\

CFSG&the Classification of Finite Simple Groups\\

$\supp\pd Y$&the support of probability distribution $\pd Y$
on a finite or countably infinite sample space, i.e.,
$\{y\st\Pro_{\pd Y}\{y\}\ne0\}$\\

$\alpha(\pd Y)$&the image of probability distribution $\pd
Y$ under function~$\alpha$\\

$\rv y_1,\dots,\rv y_n\rvd\pd Y$&means that $\rv
y_1,\dots,\rv y_n$ are independent random variables
distributed according to probability distribution~$\pd Y$\\

$\pd Y^n$&the distribution of a random variable $(\rv
y_1,\dots,\rv y_n)$, where $\rv y_1,\dots,\rv y_n\rvd\pd Y$
for probability distribution~$\pd Y$\\

$K$&an infinite subset of~$\N$\\

$D$&a subset of $\B^*$\\

$(\pd D_k\st k\in K)$&a polynomial-time samplable (when the
indices are represented in unary) probability ensemble
consisting of distributions on~$D$\\

$(D_k\st k\in K)$&a family of nonempty subsets of $D$ such
that there exists a polynomial $\theta$ satisfying
$D_k\subseteq\B^{\le\theta(k)}$ for all $k\in K$\\

$\negl$&an unspecified negligible function on~$K$\\

$\sigma$&a function from a subset of $\B^*$
onto~$F_\infty(\cl V)$\\

$\Lambda(H,\cl V,\sigma,g)$&$=\{(t,u)\in(\dom\sigma)^2\st
\sigma(t),\sigma(u)\in F_m(\cl
V)\cs\sigma(t)\ne\sigma(u)\cs\sigma(t)(g)=\sigma(u)(g)\}$,
where $H\in\cl V$ and $g\in H^m$ ($m\in\N\setminus\{0\}$);
in other words, $\Lambda(H,\cl V,\sigma,g)$ is the union of
the sets $\sigma^{-1}(v)\times\sigma^{-1}(w)$ over all
$v,w\in F_m(\cl V)$ such that $v\ne w$ and $v(g)=w(g)$\\

$\SLP_{\cl V}$, $\SLP$&the function that provides the
representation of elements of $F_\infty(\cl V)$ by
straight-line programs (see Example~\ref{ex:SLPrepres})\\

$A^H$&a black-box $\Omega$-algebra algorithm $A$ performing
a computation in black-box $\Omega$-algebra $H$ and hence
using an $\Omega$-oracle for $H$ (also in the quantum
computation model)\\

$A^{H,O}$&means that algorithm $A^H$ has access to
additional oracle~$O$\\

$Q_n$&the state space of $n$ qubits\\

$Q_n^{\otimes m}$&the $m$th tensor power of~$Q_n$
($m\in\N\setminus\{0\}$)\\

$\ket{y_1}\dots\ket{y_m}$&the state of a system of $m\ge1$
quantum registers, each consisting of $n$ qubits, when its
registers are in states $\ket{y_1},\dots,\ket{y_m}\in Q_n$
(of course, $\ket{y_1}\dots\ket{y_m}=\ket{y_1}\otimes\dots
\otimes\ket{y_m}\in Q_n^{\otimes m}$)\\

$W[i_1,\dots,i_r]$&the unitary operator on $Q_n^{\otimes m}$
(considered as the state space of the quantum system
mentioned in the previous entry) acting as unitary operator
$W\colon Q_n^{\otimes r}\to Q_n^{\otimes r}$ on the system
of registers with numbers $i_1,\dots,i_r$ (taken in this
order; here $i_1,\dots,i_r$ are distinct) and leaving all
other registers unchanged\\

$\CNOT_n$&the unitary operator on $Q_n^{\otimes2}$ such that
$\CNOT_n(\ket v\ket w)=\ket v\ket{v\oplus w}$ for all
$v,w\in\B^n$\\

$H|_\Psi$&the $\Psi$-reduct of $\Omega$-algebra $H$, where
$\Psi\subseteq\Omega$\\

$\cl V|_\Psi$&the variety of $\Psi$-algebras generated by
the $\Psi$-reducts of all $\Omega$-algebras in $\cl V$,
where $\Psi\subseteq\Omega$
\end{longtable}}
\end{document}